\documentclass{article}

    \PassOptionsToPackage{numbers, compress}{natbib}
 \usepackage[preprint]{neurips_2026}

\usepackage[utf8]{inputenc} 
\usepackage[T1]{fontenc}    
\usepackage{hyperref}       
\usepackage{url}            
\usepackage{booktabs}       
\usepackage{amsfonts}       
\usepackage{nicefrac}       
\usepackage{microtype}      
\usepackage{xcolor}         

\usepackage{bm}
\usepackage[font=small,labelfont=bf]{caption}
\usepackage{thm-restate}
\usepackage{amsmath}
\usepackage{mathrsfs}
\usepackage{xspace}
\usepackage{enumitem}
\usepackage{amsthm}

\usepackage{nicefrac}       
\usepackage{microtype}      
\usepackage{xcolor}         
\usepackage{amssymb}
\usepackage{esvect}
\usepackage{amsmath}
\usepackage{mathtools}
\usepackage{cleveref}
\usepackage{tikz}
\usepackage{pgfplots}
\usepackage{comment}
\usepackage{dsfont}
\usepackage{subcaption}
\usetikzlibrary{decorations.pathreplacing}

\usepackage{algorithm}
\usepackage[noend]{algpseudocode}

\definecolor{mygreen}{rgb}{0.0, 0.5, 0.0}
\definecolor{myorange}{rgb}{0.55, 0.62, 1}

\definecolor{niceRed}{RGB}{190,38,38}
\definecolor{Red2}{RGB}{219, 50, 54}
\definecolor{mgreen}{RGB}{160, 200, 140}
\definecolor{blueGrotto}{RGB}{5,157,192}
\definecolor{limeGreen}{HTML}{81B622}
\definecolor{myellow}{rgb}{0.88,0.61,0.14}
\definecolor{darkGreen}{HTML}{2E8B57}
\definecolor{navyBlueP}{HTML}{03468F}
\definecolor{Sepia}{HTML}{7F462C}
\definecolor{red2}{HTML}{1F462C}
\definecolor{orange2}{HTML}{FF8000}
\definecolor{mgray}{HTML}{ABB3B8}
\definecolor{lgray}{HTML}{E5E8E9}
\definecolor{myPurple}{RGB}{175,0,124}
\definecolor{mypurple2}{rgb}{0.8,0.62,1}
\definecolor{royalBlue}{HTML}{057DCD}
\definecolor{mpink}{HTML}{FC6C85}
\definecolor{lblue}{RGB}{74,144,226}
\definecolor{peagreen}{RGB}{152,193,39}
\definecolor{typ_navy}{HTML}{001f3f}
\definecolor{typ_blue}{HTML}{0074d9}
\definecolor{typ_aqua}{HTML}{7fdbff}
\definecolor{typ_teal}{HTML}{39cccc}
\definecolor{typ_eastern}{HTML}{239dad}
\definecolor{typ_purple}{HTML}{b10dc9}
\definecolor{typ_fuchsia}{HTML}{f012be}
\definecolor{typ_maroon}{HTML}{85144b}
\definecolor{typ_red}{HTML}{ff4136}
\definecolor{typ_orange}{HTML}{ff851b}
\definecolor{typ_yellow}{HTML}{ffdc00}
\definecolor{typ_olive}{HTML}{3d9970}
\definecolor{typ_green}{HTML}{2ecc40}
\definecolor{typ_lime}{HTML}{01ff70}
\definecolor{newgreen}{HTML}{83c702}
\definecolor{newpurp}{RGB}{97,96,121}
\definecolor{Andrea}{RGB}{204, 0, 153}

\hypersetup{
	colorlinks = true,
	linkcolor = typ_blue,
	citecolor = darkGreen,
	linktocpage = true,
	urlcolor = darkGreen
}

\newtheorem{lemma}{Lemma}
\newtheorem{remark}{Remark}
\newtheorem{claim}{Claim}

\newtheorem{theorem}{Theorem}

\newcommand{\E}{\mathbb{E}}

\newcommand{\Reals}{\mathbb{R}}
\newcommand{\Naturals}{\mathbb{N}}

\newcommand{\argmin}{\operatornamewithlimits{argmin}}
\newcommand{\argmax}{\operatornamewithlimits{argmax}}

\newcommand{\Prob}{\mathbb{P}}

\newcommand{\xhdr}[1]{\vspace{1mm} \noindent{\bf #1.}\hspace{1mm}}

\makeatletter
\newcommand{\mapstoto}{\mathpalette\@mapstoto\relax}
\newcommand*{\@mapstoto}[2]{%
    \mathrel{%
      \vcenter{%
         \vbox{%
            \baselineskip\z@skip
            \lineskip\z@
            \ialign{##\cr$#1\mapstochar\varrightarrow$\cr
            $#1\mapstochar\varrightarrow$\cr}%
         }%
      }%
   }%
}
\makeatother

\newcommand{\opt}{\textup{\textsc{Opt}}}
\newcommand{\I}{\mathfrak{I}}

\newcommand{\ie}{\emph{i.e.},\xspace}
\newcommand{\eg}{\emph{e.g.},\xspace}

\newcommand{\DeltaF}{\|F_{a_1} - F_{a_2} \|_1}
\newcommand{\lazy}{\theta_{\textup{lazy}}}

\crefname{line}{line}{lines}
\Crefname{line}{Line}{Lines}

\definecolor{niceRed}{RGB}{190,38,38}
\definecolor{Red2}{RGB}{219, 50, 54}
\definecolor{mgreen}{RGB}{160, 200, 140}
\definecolor{blueGrotto}{RGB}{5,157,192}
\definecolor{limeGreen}{HTML}{81B622}
\definecolor{myellow}{rgb}{0.88,0.61,0.14}
\definecolor{darkGreen}{HTML}{2E8B57}
\definecolor{navyBlueP}{HTML}{03468F}
\definecolor{Sepia}{HTML}{7F462C}
\definecolor{red2}{HTML}{1F462C}
\definecolor{orange2}{HTML}{FF8000}
\definecolor{mgray}{HTML}{ABB3B8}
\definecolor{lgray}{HTML}{E5E8E9}
\definecolor{myPurple}{RGB}{175,0,124}
\definecolor{mypurple2}{rgb}{0.8,0.62,1}
\definecolor{royalBlue}{HTML}{057DCD}
\definecolor{mpink}{HTML}{FC6C85}
\definecolor{lblue}{RGB}{74,144,226}
\definecolor{peagreen}{RGB}{152,193,39}
\definecolor{typ_navy}{HTML}{001f3f}
\definecolor{typ_blue}{HTML}{0074d9}
\definecolor{typ_aqua}{HTML}{7fdbff}
\definecolor{typ_teal}{HTML}{39cccc}
\definecolor{typ_eastern}{HTML}{239dad}
\definecolor{typ_purple}{HTML}{b10dc9}
\definecolor{typ_fuchsia}{HTML}{f012be}
\definecolor{typ_maroon}{HTML}{85144b}
\definecolor{typ_red}{HTML}{ff4136}
\definecolor{typ_orange}{HTML}{ff851b}
\definecolor{typ_yellow}{HTML}{ffdc00}
\definecolor{typ_olive}{HTML}{3d9970}
\definecolor{typ_green}{HTML}{2ecc40}
\definecolor{typ_lime}{HTML}{01ff70}
\definecolor{newgreen}{HTML}{83c702}
\definecolor{newpurp}{RGB}{97,96,121}
\definecolor{Andrea}{RGB}{204, 0, 153}

\title{Optimal Rates for Single-Dimensional Online Contract-Design with Binary Actions}
\title{Regret Minimization in Single-Dimensional Contract-Design with Binary Actions}

\author{%
  Riccardo Poiani \\
  Bocconi University \\
  \texttt{riccardo.poiani@unibocconi.it}\\
  \And
  Martino Bernasconi \\
  Bocconi University \\
  \texttt{martino.bernasconi@unibocconi.it} \\
  \AND
  Andrea Celli \\
  Bocconi University \\
  \texttt{andrea.celli2@unibocconi.it} \\
}

\begin{document}

\maketitle

\begin{abstract}

We study principal-agent problems in which a principal commits to an outcome-dependent payment scheme (\ie a \emph{contract}) in order to induce an agent to take a costly action leading to a favorable outcome. We consider the online extension of the classical (one-shot) principal-agent problem, in which the principal repeatedly interacts with agents by proposing contracts over multiple rounds. The principal has no information about the agents and, crucially, does not observe their actions. As a result, the principal must learn an optimal contract using only the realized outcomes observed at each round.
%
%
We focus on the setting with binary actions and single-dimensional agent types, where the agent's private type represents their cost per unit-of-effort.
For adversarial-type sequences, we provide tight $\Theta(T^{2/3})$ regret guarantees. Remarkably, this rate is completely independent of the number of outcomes $m$. The upper bound is based on two key components: 1) a reduction to a one-dimensional threshold optimization problem and 2) a non-uniform discretization to handle the non-Lipschitz nature of the problem. 
%
%
Moreover, in the case of a single (fixed) hidden type, we show that it is possible to improve the rates and provide a tight $\widetilde\Theta(\sqrt{T})$ regret bound. Our algorithm is based on an explore-then-commit strategy where we first approximately learn the hidden type via a stochastic binary search, and then we commit to a ``robustified'' near-optimal contract.
\end{abstract}

\section{Introduction}

Principal-agent problems model scenarios in which a principal seeks to influence the behavior of a self-interested agent capable of performing costly actions. The principal commits to an outcome-dependent payment rule (\ie a \emph{contract}), after which the agent selects one of the available actions. In the standard hidden-action model, the principal does not observe the agent's action directly, but only a stochastic outcome whose distribution depends on the chosen action and determines the principal's reward. The principal's objective is then to design a contract that incentivizes the agent to take actions leading to favorable outcomes. 
%
This model is central to modern economic theory and has recently given rise to a thriving body of research studying it through a computational lens. Such interest is motivated by applications like online labor platforms \cite{kaynar2023estimating}, strategic data provisioning \cite{cai2015optimum}, pay-for-performance healthcare \cite{bastani2016analysis,bastani2019evidence}, blockchain \cite{cong2019blockchain} and delegation to AI agents \cite{hadfield2019incomplete,saig2023delegated,saig2024incentivizing,velasco2025your,saig2026adaptive}. 

We study the generalization of the classical (single-round) model, in which the principal faces uncertainty about the agents' types and must learn an optimal contract from data. In particular, the principal repeatedly interacts with agents over multiple rounds. At each round, the principal commits to a new contract and observes the outcome produced by the agent's response. The goal of the principal is minimizing the regret with respect to the best contract in hindsight (\ie given knowledge of the sequence of agent types). 
The best known regret guarantees for general online contract design problems are nearly linear \citep{zhu2022online} or exponential in the instance size \citep{bacchiocchi2023learning}. Stronger sublinear-regret guarantees are known under additional structural assumptions, such as regularity conditions relating action costs and outcome distributions, bounded-density assumptions, or restrictions on the class of admissible contracts (see, e.g., \cite{zhu2022online,chen2024bounded,bernasconi2025single} and the discussion in \Cref{sec:related}). In this work, we ask whether there are natural settings in which one can obtain tight sublinear regret for general classes of contracts, \emph{without} imposing further assumptions on the type or outcome distributions.

Motivated by known hardness results for the Bayesian setting with multi-dimensional types \cite{castiglioni2022bayesian,guruganesh2021contracts,castiglioni2023designing}, we study the single-dimensional type setting introduced by \citet{alon2021contracts}, in which the agent's type is a scalar representing a cost per unit of effort (\ie a parameter that scales the cost of the chosen action). 
We consider the case in which agents have binary actions, so they can either exert effort or not. This is a canonical model in economics, capturing the standard ``work / shirk'' formulation of moral hazard (see, \eg \cite{ho2016adaptive,babaioff2012combinatorial,dutting2023multi,feldman2025combinatorial}).
A useful way to interpret this model is through a simple delegated-learning example in the spirit of \citet{saig2023delegated}, where a principal outsources the training of a classifier to a provider. The provider can either invest \emph{high effort} ($a_1$), for example by performing extensive data annotation and hyperparameter tuning, or training on a more expensive pipeline, or choose \emph{low effort} ($a_2$), corresponding to a cheaper baseline procedure. The provider's private type $\theta$ represents its cost per unit-of-effort, so that the cost of action $a$ is $\theta \cdot C_a$. After training, the principal does not observe which pipeline was actually used, but only a stochastic performance outcome (\eg validation accuracy), whose distribution depends on the chosen action. 
A second example is a crowdsourcing or data-labeling platform. In each round, a new worker arrives to complete a task, and the platform offers an outcome-contingent payment rule, for example, a base payment together with a bonus tied to an observable quality signal  (\eg agreement with other workers, expert review, or the performance of a downstream model trained on the collected data). The worker can either exert high effort, corresponding to careful completion of the task, or low effort, corresponding to a cheap but lower-quality submission. Again, the platform does not observe the worker's effort directly, but only the realized quality outcome, whose distribution depends on the chosen action. 
In both settings, a contract specifies a payment as a function of the realized performance outcome. In this binary-action setting, every contract trades off incentives and payments by determining which types find it worthwhile to choose the high-effort action, thereby inducing a threshold on the hidden type space. Learning a good contract then amounts to learning the right threshold from outcome-only feedback.


\subsection{Our Results and Techniques}

\xhdr{Main Results} We establish a tight $\Theta(T^{2/3})$ regret bound when the sequence of agent types is chosen adversarially (\Cref{sec:adversarial}). We then turn to the setting in which the agent type is fixed throughout the interaction (though it is hidden) and show that the optimal rate improves to $\widetilde{\Theta}(\sqrt{T})$ (\Cref{sec:single}). 

\xhdr{Key Features of Our Bounds} A notable feature of these results is that, unlike previous guarantees, they are completely independent of the number of outcomes $m$. This is particularly desirable in environments with large outcome spaces, where existing bounds may deteriorate substantially with $m$ (see, \eg the $\widetilde{O}(T^{1-1/(2m+1)})$ bound by \citet{zhu2022online}). Our guarantee for adversarial-type sequences also improves over the $\widetilde{O}(m^2 T^{4/5})$ bound of \citet{bacchiocchi2023learning} (specialized to the two-action setting), not only in its dependence on the horizon $T$, but also in its dependence on the instance parameters, since our rate does not scale with $m$ at all.
Moreover, our results are surprising if we view online contract design as a special case of repeated Stackelberg games with myopic agents \cite{letchford2009learning}. Indeed, in general, known hardness results show that there exist instances for which learning an optimal commitment requires a number of samples exponential in the number of actions of the leader (\ie the principal in our terminology) \cite{peng2019learning}. In contrast, we can exploit the specific structure of our problem to obtain sublinear regret rates, despite the principal having a continuous action space and observing only outcome feedback.
Finally, our bounds provide the first evidence that no-regret learning is possible in the single-dimensional-type setting without bounded-density assumptions, addressing a question recently posed by \citet{bernasconi2025single} in the special case of binary actions.

\xhdr{Techniques} We briefly highlight here the main technical components of our analysis, and defer a broader discussion to the later sections.
\begin{itemize}[leftmargin=1.5em]
\item The first key idea to obtain the $O(T^{2/3})$ upper bound is showing that, although contracts live in an $m$-dimensional space, the search for the best fixed contract can be reduced to a one-dimensional optimization problem over the expected payment gap between the two actions. The resulting one-dimensional objective is not $O(1)$-Lipschitz, so standard uniform discretization arguments fail. To overcome this, we develop a non-uniform discretization tailored to the geometry of the linear program that defines the cheapest contract implementing a given incentive level. This yields a discretization with only $O(\varepsilon^{-1})$ relevant points while preserving enough control on the discretization error to obtain the desired regret bound.

\item The lower bound shows that every algorithm suffers regret $\Omega(T^{2/3})$ under stochastic type generation with an absolutely continuous type distribution. The main challenges to obtain this result are the learner acts in a continuous contract space and the observed feedback does not coincide with the principal's utility. To overcome this, we construct a family of carefully perturbed continuous instances in which the utility improves only on a small interval of contracts while the induced outcome distributions remain hard to distinguish, and then apply a change-of-measure argument directly in the continuous action space.

\item Finally, in order to obtain the $\widetilde{O}(\sqrt{T})$ upper bound in the case of a single fixed type, we design an algorithm exploiting the binary-action structure to perform a stochastic binary search over the unknown threshold induced by the hidden type. The main challenge is to design exploratory contracts that are informative enough to localize the threshold while remaining close enough to optimal to keep regret under control. We address this by combining threshold-testing contracts with a robustification step in the exploitation phase.
\end{itemize}

\subsection{Related Work}\label{sec:related}

In the following, we review studies most directly related to our work. For a comprehensive overview of algorithmic contract theory, readers are referred to the recent survey by \citet{dutting2024algorithmic}.

\xhdr{Typed Agents} Contracts for typed agents have been studied in both multi- and single-dimensional models. In the multi-dimensional setting, where the costs and outcome distribution can depend arbitrarily on the type, a series of works showed that the problem of approximating optimal (deterministic) contracts is APX-hard \cite{castiglioni2022bayesian,guruganesh2021contracts,castiglioni2023designing}. The single-dimensional model admits stronger positive results. In particular, \citet{alon2021contracts,alon2023bayesian} show tractability for a constant number of actions, and also obtain positive approximation guarantees for simple contracts. Moreover, \citet{bernasconi2025single} showed that the problem admits an additive PTAS. However, for multiplicative approximations, the single-dimensional problem remains as hard as the general multi-dimensional case \cite{castiglioni2025reduction}. 

\xhdr{Online Contract Design} The problem of learning optimal contracts was introduced by \citet{ho2016adaptive} and later extended in \citet{cohen2022learning}.
%
%
The work most closely related to ours on learning optimal contracts in general principal-agent settings is \citet{zhu2022online}. They provide an algorithm guaranteeing cumulative regret of $\widetilde{O}(\sqrt{m}T^{1-1/(2m+1)})$. Moreover, they show that it is possible to learn linear contracts (which are a weaker class of contracts in which the payment is proportional to the principal's reward) incurring regret of order $\widetilde \Theta(T^{2/3})$. They also provide an almost-matching lower bound of $\Omega(T^{1-1/(2m+1)})$ for the case of general bounded contracts. Their lower-bound construction employs instances with a number of actions exponential in $m$. The question of whether it is possible to do better in settings with small action spaces was answered by \citet{bacchiocchi2023learning}, who showed that for bounded contracts it is possible to achieve regret $\widetilde{O}(m^n\,T^{4/5})$, which depends polynomially on the instance size if the number of actions $n$ is constant. 
Moreover, for linear contracts, \citet{bacchiocchi2025regret} showed that $\widetilde{O}(\sqrt{ndT})$ regret is attainable in settings with a finite number of types $d$ and a number of actions $n$ that is small relative to $T$ (\ie $n\le T^{1/3}$).

\xhdr{Learning under Additional Assumptions} In order to circumvent the negative results of \citet{zhu2022online}, subsequent works have explored learning optimal contracts under additional assumptions. \citet{chen2024bounded} show that if the problem satisfies First-Order Stochastic Dominance (FOSD) and the Concavity of Distribution Function Property (CDFP), then it is possible to achieve regret of order $\widetilde{O}(m^{11/21}\cdot T^{20/21}\cdot \log(1/\delta))$.
%
%
Recently, \citet{bernasconi2025single} studied a single-dimensional-type setting similar to ours but requiring types to be generated according to a distribution admitting density bounded by some constant $\beta$. Under this assumption, they  provide a learning algorithm with regret of order $\widetilde{O}(\beta\, \text{poly}(n,m)\, T^{3/4})$. Here, the bounded-density assumption plays a crucial role, as it provides a way to control the discretization error and thereby enables a direct reduction from the continuous problem to a discrete one. 

\xhdr{Other Feedback Models} \citet{dutting2023optimal} study a stronger feedback model where the principal can observe their expected utility instead of a sampled outcome. In this setting, linear contracts can be learned, incurring regret at most $O(\log\log T)$. 
%
Finally, we mention the recent work of \citet{dutting2025pseudo} which uses the notion of pseudo-dimension to bound the sample complexity of approximately optimal contracts. While providing fundamental barriers to learnability in the case of unbounded contracts, they provide polynomial sample complexity results for bounded contracts, which translate to a regret of order $\widetilde{O}(\sqrt{mT})$ under a full feedback model (i.e., the learner observes the type after committing to a contract). 

Our problem also shares connections with the more general problem of learning an optimal commitment, in repeated Stackelberg games \cite{letchford2009learning,peng2019learning,blum2014learning,balcan2015commitment,balcan2025nearly} and online Bayesian persuasion \cite{castiglioni2020online,castiglioni2021multi,bernasconi2023optimal}.
\section{Preliminaries}\label{sec:setting}


The single-dimensional Bayesian principal-agent problem is defined as follows.
There is a set of outcomes $\Omega$ of size $|\Omega|=m$. Each possible outcome $\omega$ provides to the principal a reward $r_\omega\in[0,1]$. There are two actions $a_1$ and $a_2$ which are the high- and low-effort actions, respectively. Each action is associated with a cost $C_a\in[0,1]$, and we can assume without loss of generality that $C_{a_1} \ge C_{a_2}$.\footnote{Without loss of generality, we assume that $C_{a_1} > 0$, otherwise the problem admits a trivial solution.} Each action $a\in A:=\{a_1,a_2\}$ is associated to a distribution over outcomes $F_a\in\Delta(\Omega)$. We denote with $F_{a,\omega}$ the probability that action $a$ leads to outcome $\omega$.
We consider the single-dimensional-type problem introduced by \citet{alon2021contracts} where the agent is associated with a hidden type $\theta\in[0,1]$, and the agent's cost of performing action $a$ is $\theta\cdot C_a$. In this setting, the uncertainty is entirely encapsulated in the private types, the outcome distributions and costs are part of the publicly known environment.


The principal's goal is to choose a contract $p \in [0,1]^m$, where each component $p_\omega$ specifies the payment made in outcome $\omega$. Based on the proposed contract $p$, the agent takes an action $a\in A$ and an outcome $\omega\sim F_{a}$ is reached. When an agent of type $\theta$ takes action $a$, the expected utility of the agent is $U^A_\theta(p,a) = \sum_{\omega \in \Omega} F_{a,\omega} p_w - \theta \cdot C_a$, while the principal has expected utility $U^P(p,a) = \sum_{\omega \in \Omega} F_{a, \omega}(r_\omega - p_{\omega})$. After committing to a contract, the principal observes only the realized outcome $\omega$, and observes neither the action $a\in A$ taken nor the type $\theta$.


When an agent of type $\theta$ observes a contract $p$, we assume that it will play an action $a \in A$ that is both (1) \emph{incentive compatible} (IC), \ie it maximizes the expected utility over actions, and (2) \emph{individually rational} (IR), \ie $U^A_\theta(p,a) \ge 0$.

To conveniently incorporate IR into IC constraints, it is standard to assume that there exists an action $a$ with cost $C_a = 0$. When multiple actions are IC, we assume that ties are broken in favor of the principal. More formally, let $\mathcal{B}(\theta,p) = \argmax_{a \in A} U^A_\theta(p,a)$ be the set of actions that are IC for type $\theta$ and contract $p$. Then, we define $b(\theta,p)$ as the action within $\mathcal{B}(\theta,p)$ that maximizes the principal's expected utility under $p$.\footnote{Whenever the utilities of $a_1$ and $a_2$ are both optimal for the principal and the agent, we assume the agent plays $a_1$.} Finally, we conclude by introducing the following notation. For $\epsilon > 0$, we denote by $\mathcal{B}_\epsilon(\theta, p)$ as the sets of actions that are $\epsilon$-IC for the agent, \ie 
\[
\mathcal{B}_\epsilon(\theta, p) = \Big\{a \in A : U^A(p,a) \ge \max_{b \in A} U^A(p,b) - \epsilon \Big\}.
\]

\textbf{Online Setting.} In this work, we focus on the online adversarial setting where the principal interacts with an arbitrary sequence of $T \in \Naturals$ agents. During each round $t \in [T]$, a type $\theta_t$ arrives, and the principal selects a contract $p_t \in [0,1]^m$. The agent $\theta_t$ best responds, \ie it plays an action $a_t \in b(\theta_t, p_t)$ and the principal observes $\omega_t \sim F_{a_t}$ and receives $r_{\omega_t} - p_{w_t}$. 
The principal's goal is to minimize the regret against the best contract in hindsight for any fixed sequence of types, \ie
\[
R_T = \sup_{\{\theta_t\}_{t\in[T]}}\sup_{p \in [0,1]^m}\left\{ \sum_{t=1}^T U^P(p, b(\theta_t,p)) - \E \left[ \sum_{t=1}^T U^P(p_t, b(\theta_t, p_t)) \right]\right\}, 
\]
where the expectation is taken with respect to the stochasticity of the algorithm. 
In the following, whenever it is clear from context, we write $U^P(p,\theta) = U^P(p,b(\theta,p))$. Moreover, for a fixed sequence of types, we define $\opt \coloneqq \sup_{p \in [0,1]^m} \sum_{t=1}^T U^P(p,\theta_t)$, omitting the dependence on the type sequence when it is clear from context. Throughout this document, we assume that $F_{a_1}^\top r > F_{a_2}^\top r$. If this condition does not hold, then there exists a trivial algorithm with zero regret.\footnote{In this case, the null contract is optimal for every agent $\theta \in [0,1]$. We prove this formally in \Cref{lemma:easy-case}.}
\section{Adversarial Types: $\Theta(T^{2/3})$ is Tight}\label{sec:adversarial}

In this section, we consider the general adversarial setting in which an arbitrary sequence of types $\{\theta_t\}_{t=1}^T$ arrives. In \Cref{sec:ub-any-seq}, we propose an algorithm that attains ${O}(T^{2/3})$ regret. Then, in \Cref{sec:lb-any-seq}, we show that this rate is tight in the worst-case sense by proving a matching lower bound. 

\subsection{The Adversarial $O(T^{2/3})$ Upper Bound}\label{sec:ub-any-seq}

In this section, we prove the following upper bound on adversarial types and provide a high-level overview of the algorithm. The formal proof is deferred to \Cref{app:proof-thm-1}.
 
\begin{theorem}\label{thm:ub1}
    There exists an algorithm for the online contract-design problem with binary actions and single-dimensional adversarial types, such that $R_T \le O(T^{2/3})$.
\end{theorem}


A notable feature of the upper bound is that it is independent of the number of outcomes $m$. This is surprising since the action space is continuous and standard discretization arguments typically incur an exponential dependence on $m$.  Our approach avoids this through two main ideas.
First, we show that the search for an optimal contract can be reduced from an $m$-dimensional space to a one-dimensional one. The key parameter is the threshold type, \ie the highest type that can still be induced to choose action $a_1$, together with the cheapest contract that implements it (see \Cref{sec:singledimprob}). This reduces the problem to optimization over the interval $[0,1]$.
Second, once we obtained the one-dimensional problem, the standard way of handling it would be to discretize the $[0,1]$ interval uniformly with step $\epsilon=T^{-1/3}$. If the underlying reward function was $O(1)$-Lipschitz, this would give a $O(\sqrt{T/\epsilon}+\epsilon T)=O(T^{2/3})$ upper bound (this also holds for one-sided Lipschitz functions \citep{dutting2023optimal}). However, when reformulated as a one-dimensional problem, the reward is not $O(1)$-Lipschitz. As a result, a clever discretization strategy is needed to avoid dependence on the number of outcomes $m$. This motivates the second key idea: a non-uniform discretization that exploits additional structure of the problem, which we discuss in more detail in \Cref{sec:discret}.

%

\subsubsection{The One-Dimensional Problem}\label{sec:singledimprob} To show that the problem can be reduced to a one-dimensional one, we start by analyzing the utility of the principal when it plays a contract $p$. To this end, we will only focus on the subset of contracts for which $F_{a_1}^\top (r-p) \ge F_{a_2}^\top (r-p)$ holds, \ie the subset where ties in $\mathcal{B}(\theta,p)$ are broken in favor of $a_1$, which yields larger utility for the principal. We denote this subset by $\mathcal{P}_1$. Indeed, we show in \Cref{lemma:adv-opt-p1} that $\opt$ can be rewritten as an optimization over $\mathcal{P}_1$ only.\footnote{From an intuitive perspective, the reason is that for contracts that do not belong to $\mathcal{P}_1$, $F_{a_2}^\top (r-p) \ge F_{a_1}^\top (r-p)$ holds. Hence, for any $\theta$, their utility is at most $F_{a_2}^\top (r-p)\le F_{a_2}^\top r$, which is always smaller or equal than the utility of the contract $\bar p = (0, \dots, 0)$, which belongs to $\mathcal{P}_1$.} The next observation is that, for a fixed $\theta \in [0,1]$, the utility of any contract in $\mathcal{P}_1$ can be rewritten as follows:
\begin{align}
    U^P(p, \theta) & = \mathds{1}\{ \theta \le \tilde \theta_{\max}(\lambda_p) \} F_{a_1}^\top (r-p) + \mathds{1}\{ \theta > \tilde \theta_{\max}(\lambda_p) \} F_{a_2}^\top (r-p)  \label{eq:main-utility-eq1} \\
    & = \mathds{1}\{ \theta \le \tilde \theta_{\max}(\lambda_p) \} ( (F_{a_1} - F_{a_2})^\top r - \lambda_p ) + F_{a_2}^\top r + \lambda_p - F_{a_1}^\top p, \label{eq:main-utility-eq2} 
\end{align}
where $\lambda_p$ is a shorthand for $(F_{a_1} - F_{a_2})^\top p$ and $\tilde \theta_{\max}(\lambda) = \lambda C_{a_1}^{-1}$. The rationale behind this rewriting is that, when the contract is $p$, all types $\theta \le \tilde \theta_{\max}(\lambda_p)$ will best respond by playing $a_1$, while types larger than $\tilde \theta_{\max}(\lambda_p)$ will play $a_2$. 
We can already see that, to induce the agent to choose action $a_1$, the principal will need to increase the expected payment gap $\lambda_p$ between the two actions.
\Cref{eq:main-utility-eq2} explicitly shows the dependence of the utility on $\lambda_p$. This is important as it allows us to reduce the problem to one dimension. Indeed, if we fix the expected difference in payment $\lambda$, \Cref{eq:main-utility-eq2} highlights how the optimal contract for that $\lambda$ is simply the one that minimizes the expected payment to $a_1$, \ie the solution of the following linear program parametrized in $\lambda$:
\begin{align}\label{eq:opt-problem-main}
f(\lambda):=\min_{p \in \mathcal{P}_1} F_{a_1}^\top p \quad \text{s.t.}\quad  (F_{a_1} - F_{a_2})^\top p = \lambda.    
\end{align}
Now we can rewrite the principal utility entirely as a function of $\lambda$ and the value of the LP $f(\lambda)$, \ie
\[
\widetilde{U}^P(\lambda, \theta) = \mathds{1}\{ \theta \le \tilde \theta_{\max}(\lambda) \} ( (F_{a_1} - F_{a_2})^\top r - \lambda ) + F_{a_2}^\top r + \lambda - f(\lambda).
\]
In other words, we restricted the principal utility to a one-dimensional function where the $\lambda$s of interests are those for which there exists $p \in \mathcal{P}_1$ that yields $\lambda$ as the expected payment gap. We show in \Cref{lemma:opt-rewrite} that the only relevant range of $\lambda$ is $[0, \bar\lambda]$ with $\bar\lambda:=(F_{a_1} - F_{a_2})^\top r$. The upper endpoint is obtained  from the definition of $\mathcal{P}_1$, while the lower endpoint is clipped to $0$. The intuitive reason is that it is never convenient to incentivize $a_2$ over $a_1$, since $a_1$ is the rewarding action.


\subsubsection{The Non-Uniform Discretization}\label{sec:discret} 

\begin{figure}
    \centering
    \scalebox{0.6}{\input{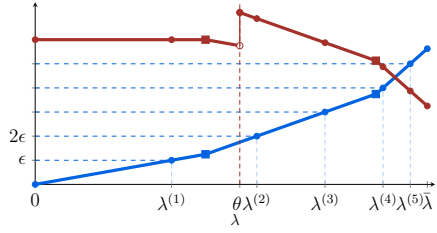}}
    \caption{Visual representation of the instance given by $4$ outcomes, $F_{a_1}=(\tfrac{1}{5},\tfrac{1}{10},\tfrac1{5},\tfrac{1}{2})$, $F_{a_2}=(\tfrac{1}{2},0,\tfrac{1}{10},\tfrac{2}{5})$, $r=(0,\tfrac12,\tfrac{4}{5},1)$ and $C_{a_1}=1$. In the plot we can see represented $f(\lambda)$ in \textbf{\textcolor[HTML]{0064E0}{blue}}, and $\lambda\mapsto U^P(\lambda,\theta)$ in \textbf{\textcolor[HTML]{A4322A}{red}} for $\theta=0.12$ and $\epsilon=0.08$. In the plot, we can also observe the type threshold $\theta$ where the agent switches from playing $a_1$ to $a_2$. Square markers indicate switches in the active constraints in the LP solution defining $f(\lambda)$, and therefore changes in its slope. The round marker indicates the result of the discretization \Cref{eq:discretization-main}.
}
    \label{fig:placeholder}
\end{figure}
At this point, ideally we would like to discretize $[0, \bar\lambda]$ with $\epsilon^{-1}$ points, and apply a no-regret adversarial bandit algorithm using as the action set the contracts yielding the discretized set of $\lambda$s. To this end, we start by analyzing the utility difference when going from $\lambda_1$ to $\lambda_2 > \lambda_1$. Precisely, given type $\theta$, one can show (\Cref{lemma:obj-diff}) that
\begin{align}\label{eq:utility-difference-main}
\widetilde{U}^P(\lambda_1, \theta) - \widetilde{U}^P(\lambda_2, \theta) \le f(\lambda_2) - f(\lambda_1).    
\end{align}
Observe that $f(\lambda)$ is the value of a linear program where $\lambda$ is the right-hand side of the equality constraint. As a result $f(\lambda_2) - f(\lambda_1)$ cannot be directly upper bounded by $c|\lambda_2 - \lambda_1|$ for some universal constant $c \in \Reals$.\footnote{Note that this is what is needed by uniform discretization to guarantee that $\epsilon^{-1}$ points approximates $\opt$ with precision $\epsilon$.} However, it is not hard to see that $f(\lambda)$ is a strictly increasing continuous function of $\lambda$, and is bounded in $[0,1]$. As a consequence, we propose the following discretization of $[0, \bar\lambda]$ that meets our requirements. 
Let $\epsilon > 0$ and consider the following sequence. Set $\lambda^{(0)} = 0$ and for all $n \ge 1$ define $\lambda^{(n+1)}$ as follows:
\begin{align}\label{eq:discretization-main}
    \lambda^{(n+1)} = \begin{cases}
        \inf\{\lambda: \lambda \ge \lambda^{(n)}, f(\lambda) = f(\lambda^{(n)}) + \epsilon \} & \textup{if~} f(\lambda^{(n)}) + \epsilon < f(\bar\lambda) \\
        \bar\lambda & \textup{otherwise} 
    \end{cases} 
    .
\end{align}
We stop the sequence as soon as $\lambda^{(n)} = (F_{a_1} - F_{a_2})^\top r$ and denote by $\Lambda^{(\epsilon)}$ the collection of points in this sequence. By taking a closer look, we can see that \Cref{eq:discretization-main} generates points in $[0, \bar\lambda]$ in a way that the distance in $f(\cdot)$ between two subsequent points is always exactly $\epsilon$.\footnote{This is true except for the last point, where the distance can actually be smaller than $\epsilon$.} Thus, by \Cref{eq:discretization-main} and \Cref{eq:utility-difference-main}, we have that $\Lambda^{(\epsilon)}$ approximates the utility of any $\lambda$ with accuracy $\epsilon$. Furthermore, since $f$ is bounded by $[0,1]$ for $\lambda\in[0,\bar\lambda]$, we have $|\Lambda^{(\epsilon)}| =O( \epsilon^{-1})$. In \Cref{fig:placeholder} we can see a visual representation of the non-uniform discretization described in this section.

\xhdr{Remark: computing $\Lambda^{(\epsilon)}$ is efficient} We observe that to compute the $(n+1)$-term of $\Lambda^{(\epsilon)}$ efficiently, we only need to solve the following linear feasibility problem in $p$ and $\lambda$: 
\begin{align}\label{eq:main-feas}
p\in \mathcal{P}_1, \lambda \in [\lambda^{(n)}, (F_{a_1} - F_{a_2})^\top r] ~\land~ (F_{a_1} - F_{a_2})^\top p = \lambda ~\land~ F_{a_1}^\top p = f(\lambda^{(n)}) + \epsilon.     
\end{align}
When a solution exists, it is equivalent to the first case of \Cref{eq:discretization-main}, and the point $\lambda$ that satisfies \Cref{eq:main-feas} is exactly $\lambda^{(n+1)}$. When it does not exist, it means $\lambda^{(n+1)} = (F_{a_1} - F_{a_2})^\top r$, and we stop the discretization procedure. Observe that the procedure that we described also yields a set of contracts whose expected payment gaps are exactly those in $\Lambda^{(\epsilon)}$. Indeed, these are the contracts satisfying \Cref{eq:main-feas}. In the following, we will denote these contracts with $\mathcal{P}^{(\epsilon)}$, and we recall that $|\mathcal{P}^{(\epsilon)}| = |\Lambda^{(\epsilon)}| \approx \epsilon^{-1}$.

\subsubsection{Putting Everything Together} At this point, we just need to combine all the pieces together. First, we compute $\Lambda^{(\epsilon)}$ and the corresponding sets of contracts $\mathcal{P}^{(\epsilon)}$, and then we run the adversarial bandit algorithm of \citet*{audibert2010regret} using as arms the contracts in $\mathcal{P}^{(\epsilon)}$. This procedure, for any $\epsilon > 0$, guarantees that
\(
R_T \le \mathcal{O}(\sqrt{T |\mathcal{P}^{(\epsilon)}|} + \epsilon T),
\)
where the first term has to be interpreted as the regret of the principal against the best contract in $\mathcal{P}^{(\epsilon)}$ while the second one, instead, comes from the fact that contracts in $\mathcal{P}^{(\epsilon)}$ approximate the optimal constant fixed in hindsight up to accuracy $\epsilon T$ (\Cref{lemma:discretization}). Indeed, this is the key property which we obtained with the discretization described above. Then, since $ |\mathcal{P}^{(\epsilon)}| \approx \epsilon^{-1}$, picking $\epsilon = T^{-1/3}$ yields the desired result.

\subsection{A Matching $\Omega(T^{2/3})$ Lower Bound}\label{sec:lb-any-seq}

In this section, we will prove that \Cref{thm:ub1} is tight up to constants. In particular, we provide a lower bound which holds even for stochastic type generation. Importantly, we show that the lower bound holds even when the type distribution is absolutely continuous (with respect to Lebesgue measure). 

The main difference from standard Multi-Armed Bandit lower bounds is that the feedback received by the learner and the reward suffered by the learner do not coincide, and that the action space of the learner is continuous.
This is a characteristic of lower bounds for similar economic models, such as dynamic pricing \citep{kleinberg2003value, babaioff2015dynamic}, bilateral trade \citep{cesa2024regret, bernasconi2024no}, and market making \citep{cesa2024market}. A principled way to deal with such problems is to reduce them to an instance with a finite set of ``arms'' and apply standard KL chain rule techniques (see \citep{slivkins2019introduction,lattimore2020learning}). This is done by carefully identifying arms that are optimal with respect to regret or information acquired. However, this would require introducing point masses into the distribution. Instead, here, to produce an absolutely continuous distribution, we do not reduce to the finite case and use the KL chain rule with continuous actions.

More formally, we create $K=\Theta(1/\epsilon)$ instances $\{\I_1,\ldots, \I_K\}$ such that
\begin{itemize}[leftmargin=0.5cm]
    \item each instance $\I_k$ is identical to a base instance $\I_0$, apart from an interval $J_k$ of size $O(\epsilon)$;
    \item the principal utility on the interval $J_k$ is $\Omega(\epsilon)$ larger than the utility outside $J_k$;
    \item the Kullback-Leibler divergence between the instance $\I_0$ and instance $\I_k$ is less than $O(\epsilon^2)$ times the number of times we played in expectation (under the base instance) in the interval $J_k$.
\end{itemize}

Then, by standard change of measure arguments, we can prove that for each algorithm, there is an instance in which the regret is at least $\Omega(\epsilon T - \epsilon^2\sqrt{T/K})$, and by choosing  $K=\Theta(1/\epsilon)$ and calibrating $\epsilon=\Theta(T^{-1/3})$ gives the desired $\Omega(T^{2/3})$ lower-bound (the full proof is in \Cref{sec:lb}).

\begin{theorem}\label{thm:lb1}
    For any algorithm, there exists an instance with stochastic types such that the expected regret is lower bounded by $\Omega(T^{2/3})$.
\end{theorem}

\section{Improved Rates Against a Single Hidden Type}\label{sec:single}

\begin{algorithm}[t]
\caption{Regret minimization against a single type}
\label{alg:e2c}
\begin{algorithmic}[1]

\Require{Threshold $\beta$, robustification $\alpha$, number of phases $L \in \Naturals$, samples per phase $n \in \Naturals$}

\If{$\| F_{a_1}- F_{a_2} \|_1 \le \beta$} 
    \State{Play $p_t = (0, \dots, 0)$ for all $t \in [T]$}\label{line:indistinguishable}
\Else
    \If{$\lazy < 1$} \Comment{Test if $\bar \theta > \lazy$}
    \State{Play any $p^\star_{\lazy}$ for $n$ rounds; observe $\{ \omega_i^{(0)} \}_{i=1}^n$; compute $Z_0 \coloneqq \frac{1}{n} \sum_{i=1}^n \mathds{1}\{ \omega_i^{(0)} \in \Omega_1 \}$}
    \If{$ Z_0 \le \tau $}
    \State{Play $\bar p = (0, \dots, 0)$ for all the remaining rounds}\label{line:null}
    \EndIf
    \EndIf
    
    \State{Set $s_1 = 0$, $e_1 = \lazy$, $\hat \theta_1 = (e_1 + s_1) / 2$} \Comment{Binary search in $[0, \lazy]$}
    \For{$l \in \{1, \dots, L-1 \}$}
    \State{Play $p^\star_{\hat \theta_l}$ for $n$ rounds; observe $\{ \omega_i^{(l)} \}_{i=1}^n$; compute $Z_l \coloneqq \frac{1}{n} \sum_{i=1}^n \mathds{1}\{ \omega_i^{(l)} \in \Omega_1 \}$ }\label{line:binarys}
    \If{$Z_l \le \tau$}
    \State{Set $s_{l+1} = \hat \theta_l$, $e_{l+1} = e_l$, $\hat{\theta}_{l+1} = (e_{l+1} + s_{l+1}) / 2$}
    \Else
    \State{Set $s_{l+1} = s_l $, $e_{l+1} = \hat \theta_{l}$, $\hat{\theta}_{l+1} = (e_{l+1} + s_{l+1}) / 2$}
    \EndIf
    \EndFor

    \State{Compute $p^{\star}_{\hat \theta_L}$ and play $\tilde p  = p^{\star}_{\hat \theta_L} + \alpha(r- p^{\star}_{\hat \theta_L})$ until the end of the horizon} \Comment{Exploit}
\EndIf

\end{algorithmic}
\end{algorithm}

In this section, we discuss how we can improve the theoretical guarantees when the principal is repeatedly playing against a fixed, but unknown, type $\bar \theta$. In \Cref{sec:ub-fixed-type} we propose an algorithm that attains $\widetilde{O}(\sqrt{T})$ regret and in \Cref{sec:lb-fixed-type} we prove a $\Omega(\sqrt{T})$ lower bound. 

\subsection{An Explore-and-Commit Strategy for $\widetilde{O}(\sqrt{T})$ Regret}\label{sec:ub-fixed-type}

We first state the main result of this section and then we explain how it is obtained. 

\begin{theorem}\label{thm:ub2}
    There exists an algorithm such that $R_T \le \widetilde{O}(\sqrt{T})$ for any sequence $\{ \bar \theta \}_{t=1}^T$ where $\bar \theta$ is any unknown type in $[0,1]$.
\end{theorem}

\xhdr{Outline} We start by observing that, if the principal knew the hidden type $\bar\theta$, then it could simply play an optimal contract for that type. Suppose, for the moment, that an oracle provided an $\epsilon$-accurate estimate $\hat\theta$ of $\bar\theta$. However, simply playing the optimal contract $p^\star_{\hat\theta}$ for $\hat\theta$ might be greatly suboptimal; indeed, even infinitesimal errors in the type $\hat\theta$ might change the action played by the agent and lead to very suboptimal outcomes. However, we can instead compute the optimal contract $p^\star_{\hat\theta}$ for type $\hat\theta$ and play its ``robustified'' version $\tilde p_{\hat\theta}= p^\star_{\hat\theta}+\alpha(r-p^\star_{\hat\theta})$ for some $\alpha>0$.
The role of robustification is to make the contract insensitive to small errors in the estimated type at the cost of only a small loss in utility. By the standard robustification argument of \citet{dutting2019simple} (see also \citet{bernasconi2024regret}), choosing $\alpha=\sqrt{\epsilon}$ guarantees $U^P(\tilde p_{\hat\theta},\theta) \ge U^P(p^\star_{\hat\theta},\theta)-2\sqrt{\epsilon}$
for nearby types (see \Cref{lemma:robust-contract}).
Therefore, an $\epsilon$-accurate estimate of $\bar\theta$ translates into $O(\sqrt{\epsilon})$ per-round exploitation regret. The goal of the algorithm is then to obtain such an estimate using a short exploration phase. To do this, we show how to build suitable threshold tests that we will use inside a stochastic binary search procedure. Crucially, the exploration contracts themselves are chosen so that each test incurs only small regret.

\xhdr{Testing a Threshold}
Fix a contract $p_\theta$ such that all types $\theta' \le \theta$ choose $a_1$ and all types $\theta' > \theta$ choose $a_2$. Playing $p_\theta$ for $n$ rounds yields samples from $F_{a_1}$ if $\bar\theta \le \theta$ and from $F_{a_2}$ otherwise. We distinguish these two cases using the statistic
\[
Z = \frac{1}{n} \sum_{i=[n]} \mathds{1}\{ \omega_i \in \Omega_1 \},
\qquad
\Omega_1 = \{ \omega \in \Omega : F_{a_1,\omega} > F_{a_2,\omega} \}.
\]
We compare $Z$ with $\tau = \sum_{\omega \in \Omega_1} (F_{a_1,\omega}+F_{a_2,\omega})/2$.
If $\bar\theta > \theta$, then $Z$ concentrates below $\tau$, and otherwise above $\tau$. We show the error probability of this test decays as $O(\exp(-n \DeltaF^2))$.\footnote{Note that the rate is independent of $m$, unlike a test based on estimating the full outcome distribution.}

\xhdr{Exploring with Low-Regret Contracts}
The contracts used for threshold tests must satisfy two properties: 1) they must not incur large regret, and 2) they must reveal whether $\bar\theta$ lies above or below the tested threshold. Consider any contract $p$, changing the type from $\theta$ to $\theta'$ can only change the principal's utility if the induced
action switches between $a_1$ and $a_2$. In that case, the loss is controlled by the distance between the two outcome distributions, \ie $|U^P(p,\theta)-U^P(p,\theta')| \le  O(\DeltaF)$. Consequently, the optimal value function is also $O(\DeltaF)$-stable across types. Indeed,
if $p^\star_{\theta'}$ is optimal for type $\theta'$, then
\[
\begin{aligned}
    \max_{p\in[0,1]^m} U^P(p,\bar\theta)
    - U^P(p^\star_{\theta'},\bar\theta)
    &\le O(\DeltaF).
\end{aligned}
\]
Hence, even if $p^\star_{\theta'}$ is not optimal for the true type
$\bar\theta$, its per-round regret against $\bar\theta$ is at most $O(\DeltaF)$. Thus, property (1) is satisfied. 
Now, let $\lazy$ be the largest type for which it is still optimal for the principal to incentivize the high-effort action $a_1$ rather than induce the low-effort action $a_2$ with the null contract.\footnote{We defer the exact expression of $\lazy$ to \Cref{eq:theta-lazy} in \Cref{app:proof-ub-fixed}. Here, we only mention that it is the solution of a linear program, and thus it can be computed efficiently.} For every $\theta>\lazy$, it is optimal to induce $a_2$, and the null contract $\bar p=(0,\dots,0)$ is optimal. 
This case is therefore easy: it is sufficient to detect that $\bar\theta>\lazy$ and commit to $\bar p$ without needing to estimate $\bar\theta$ more accurately.
For every $\theta \le \lazy$, instead, the optimal contract $p^\star_\theta$ can be used for exploration because it implements a sharp threshold at $\theta$. Indeed, $p^\star_\theta$ is the cheapest contract that induces type $\theta$ to play $a_1$, and hence satisfies $(F_{a_1}-F_{a_2})^\top p^\star_\theta = \theta C_{a_1}$.
Then, any for type $\theta' \le \theta$ action $a_1$ is IC, and it is not IC for $\theta' > \theta$. Thus, when $p^\star_\theta$ is played, the observed outcomes are drawn from $F_{a_1}$ if $\bar\theta \le \theta$ and from $F_{a_2}$ if $\bar\theta > \theta$. This makes $p^\star_\theta$ an informative threshold test for binary search, and thus, property (2) is satisfied.

\xhdr{Algorithm}
\Cref{alg:e2c} first handles the ``nearly indistinguishable'' case $\DeltaF \le \beta$ (Line \ref{line:indistinguishable}). Since the null contract is then suboptimal by at most $O(\DeltaF)$ per round, playing it throughout gives regret $O(\beta T)$, and by choosing $\beta = T^{-1/2}$ we have the desired $\sqrt{T}$ bound. 
When $\DeltaF > \beta$, the algorithm first tests whether $\bar\theta > \lazy$ by playing $p^\star_{\lazy}$ for $n$ rounds. If the test indicates $\bar\theta > \lazy$, it commits to the null contract, which is optimal for all such types (Line \ref{line:null}). Otherwise, it performs $L-1$ phases of stochastic binary search on $[0,\lazy]$. In phase $l$, the midpoint $\hat\theta_l$ is tested by playing $p^\star_{\hat\theta_l}$ for $n$ rounds and comparing the statistic $Z_l$ with $\tau$ (Line \ref{line:binarys}). After $L$ phases, the algorithm obtains an estimate $\hat\theta_L$ with error at most $2^{-L}$ on the event that all tests are correct. It then commits to the robustified contract $\tilde p=p^\star_{\hat\theta_L}+\alpha(r-p^\star_{\hat\theta_L})$, this ensures that we are playing a contract that actually yields good utility for types close to $\hat \theta_L$.

\xhdr{Proof Sketch} 
We sketch the regret bound for the main case $\bar\theta \le \lazy$ and $\DeltaF \ge T^{-1/2}$, the other cases are simpler (the complete proofs can be found in \Cref{app:proof-ub-fixed}). During the exploration phase, the algorithm plays either the null contract or contracts $p^\star_\theta$ that are optimal for some type. By the suboptimality bound above, the regret accumulated during exploration is $O(nL\DeltaF)$. For the exploitation phase, let
$\mathcal{E} \coloneqq \{\bar\theta \in [s_L,e_L]\}$ be the event that the binary search has localized the true type. On $\mathcal{E}$, we have $|\bar\theta-\hat\theta_L|\le 2^{-L}$, and robustification yields per-round regret of
$O(2^{-L}\alpha^{-1}+\alpha)$. On the complement $\mathcal{E}^\complement$, a union bound over the $L$ threshold tests gives $\mathbb{P}(\mathcal{E}^\complement)\le O(L\exp(- n\DeltaF^2))$. Even on $\mathcal{E}^\complement$, the committed contract is close to an optimal contract for some type, so the per-round regret is bounded by
$O\left((\DeltaF+\alpha) L\exp(- n\DeltaF^2)\right)$.
Combining the two terms we obtain the following regret upper bound
\[
R_T={O}\left(nL\DeltaF+ T\left[2^{-L}\alpha^{-1}+\alpha+(\DeltaF+\alpha)L\exp(-n\DeltaF^2)\right]\right).
\]
Choosing $\alpha \approx \sqrt{2^{-L}}$, $L \approx \log T$, $n \approx \DeltaF^{-2}\log(T\DeltaF^2)$ proves \Cref{thm:ub2}.

\subsection{$\Omega(\sqrt{T})$ is Unavoidable Against a Fixed Type}\label{sec:lb-fixed-type}

We complement the upper bound with a matching lower bound. We construct instances that allow us to reduce the problem to a multi-armed bandits with two arms. Precisely, we construct two instances such that, in each instance, the optimal arm has $\Omega(\epsilon)$ higher utility than it does in the other, while the KL divergence between the feedback distributions of the two instances is at most $O(\epsilon^2)$. A change-of-measure argument then yields the result. The complete proof is reported in \Cref{app:lbsingle}.

\begin{restatable}{theorem}{LBsingle}\label{th:LB_single}
    Any (possibly randomized) algorithm suffers at least regret $R_T=\Omega(\sqrt{T})$.
\end{restatable}
\section{Future Work}\label{sec:conclusion}

%
Several questions on the learnability of contracts remain open. A first direction is to strengthen lower bounds for more general online contract-design settings, especially by constructing hard instances whose dependence on the instance size is only polynomial. Another important direction is to move beyond the binary-action case: with multiple actions, the type space induces a richer arrangement of incentive regions, and it is unclear whether similar outcome-independent regret guarantees are possible. More broadly, it would be interesting to study intermediate regimes between adversarial and fixed types, such as slowly drifting types. 

\begin{ack}
This work was partially funded by the European Union.
Views and opinions expressed are however those of the author(s) only and do not necessarily reflect those of the European Union or the European Research Council Executive Agency. Neither the European Union nor the granting authority can be held responsible for them. 

This work is supported by an ERC grant (Project 101165466 — PLA-STEER)
\end{ack}

\newpage
\bibliographystyle{plainnat}
\bibliography{bibliography}

\newpage
\appendix
\section{Proof of \Cref{thm:ub1}}\label{app:proof-thm-1}

In this section we prove \Cref{thm:ub1}. This section is structured as follows. First, in \Cref{app:sec-one-dim}, we show that, to reach optimality, we can focus on a one-dimensional problem. Then, in \Cref{app:sec-disc}, we exploit this result using an ad-hoc non uniform discretization. Finally, in \Cref{app:regret-anal-thm1}, we combine these results and prove \Cref{thm:ub1}.

\subsection{A one-dimensional problem}\label{app:sec-one-dim}

For convenience, we define $\mathcal{P}_1, \mathcal{P}_2$ as follows:
\begin{align*}
& \mathcal{P}_1 = \{ p \in [0,1]^m : (F_{a_1} - F_{a_2})^\top p \le (F_{a_1} - F_{a_2})^\top r \} \\
& \mathcal{P}_2 = \{ p \in [0,1]^m : (F_{a_1} - F_{a_2})^\top p \ge (F_{a_1} - F_{a_2})^\top r \}.
\end{align*}

Then, in the following Lemma, we first show that $\opt$ can be reduced to an optimization problem over $\mathcal{P}_1$.

\begin{lemma}[Optimize only over $\mathcal{P}_1$]\label{lemma:adv-opt-p1}
    Suppose that $F_{a_1}^\top r > F_{a_2}^\top r$, then, for any sequence $\{ \theta_t \}_{t=1}^T$ it holds that:
    \[
    \opt = \sup_{p \in \mathcal{P}_1} \sum_{t=1}^T U^P(p, b(\theta_t, p)).
    \]
\end{lemma}
\begin{proof}
    Recall that $\mathcal{P}_2 = \{ p \in [0,1]^m : F_{a_2}^\top (r-p) \ge F_{a_1}^\top (r-p) \}$. Thus, take any $p \in \mathcal{P}_2$ and any $\theta \in [0,1]$ it holds that:
    \begin{align}\label{eq:opt-p1-eq-1}
        U^P(p, b(\theta,p)) & = \mathds{1} \{ a_1 \in b(\theta, p)   \} F_{a_1}^\top (r-p) + \mathds{1} \{ b(\theta, p) = a_2  \} F_{a_2}^\top (r-p) \le  F_{a_2}^\top r.
    \end{align}
    Next, observe that $\bar p = (0, \dots, 0)$ belongs to $\mathcal{P}_1$ since $F_{a_1}^\top (r - \bar p ) = F_{a_1}^\top r > F_{a_2}^\top r = F_{a_2}^\top (r- \bar p)$ and furthermore:
    \begin{align}\label{eq:opt-p1-eq-2}
        U^P(\bar p, b(\theta, \bar p)) & = \mathds{1} \{ a_1 \in b(\theta,\bar p)   \} F_{a_1}^\top r + \mathds{1} \{ b(\theta, \bar p) = a_2  \} F_{a_2}^\top r > F_{a_2}^\top r.
    \end{align}
    Chaining these results, we have that:
    \begin{align*}
        \opt & = \sup_{p \in [0,1]^m} \sum_{t=1}^T U^P(p, b(\theta_t, p)) \\
        & = \max \left\{ \sup_{p \in \mathcal{P}_1}  \sum_{t=1}^T U^P(p, b(\theta_t, p)) , \sup_{p \in \mathcal{P}_2}  \sum_{t=1}^T U^P(p, b(\theta_t, p)) \right\}  \\
        & \le  \max \left\{ \sup_{p \in \mathcal{P}_1}  \sum_{t=1}^T U^P(p, b(\theta_t, p)) , F_{a_2}^\top r \right\} \tag{\Cref{eq:opt-p1-eq-1}} \\
        & = \sup_{p \in \mathcal{P}_1}  \sum_{t=1}^T U^P(p, b(\theta_t, p)) \tag{\Cref{eq:opt-p1-eq-2}} \\
        & \le \opt .
    \end{align*}
    Hence, it must hold that $\opt = \sup_{p \in \mathcal{P}_1}  \sum_{t=1}^T U^P(p, b(\theta_t, p))$, which concludes the proof.
\end{proof}


Next, we start by rewriting the utility of contracts that belong to $\mathcal{P}_1$. 

\begin{lemma}[Utility rewriting]\label{lemma:utility}
    Consider $p \in \mathcal{P}_1$ and let $\lambda_p = (F_{a_1} - F_{a_2})^\top p$. Then, for any $\theta \in [0,1]$ it holds that:
    \[
    U^P(p, b(\theta, p)) = \mathds{1}\{ \theta \le \tilde \theta_{\max}( \lambda_p )\} \left( (F_{a_1} - F_{a_2})^\top r - \lambda_p  \right) + F_{a_2}^\top r + \lambda_p - F_{a_1}^\top p. 
    \]
\end{lemma}
\begin{proof}
    Observe that $a_1 \in \mathcal{B}(\theta, p)$ holds if $F_{a_1}^\top p - \theta C_{a_1} \ge F_{a_2}^\top p$, \ie if $\theta \le \tilde \theta_{\max}( \lambda_p )$. At equality, ties break in favor of the principal and therefore, since $p \in \mathcal{P}_1$, at $\theta = \tilde \theta_{\max} ( \lambda_p )$ the agent will play action $a_1$. Thus, we have that:
    \begin{align*}
        U^P(p, b(\theta,p)) & =  \mathds{1} \{ a_1 \in b(\theta, p)   \} F_{a_1}^\top (r-p) + \mathds{1} \{ a_2=b(\theta, p)   \} F_{a_2}^\top (r-p)  \\
        & = \mathds{1} \{ \theta \le \tilde \theta_{\max}( \lambda_p )  \} F_{a_1}^\top (r-p) + \mathds{1} \{ \theta > \tilde \theta_{\max}( \lambda_p ) \} F_{a_2}^\top (r-p) \\
        & =  \mathds{1} \{ \theta \le \tilde \theta_{\max}( \lambda_p )  \} \left( F_{a_1}^\top (r-p) - F_{a_2}^\top (r-p) \right) + F_{a_2}^\top (r-p) \\
        & =  \mathds{1} \{ \theta \le \tilde \theta_{\max}( \lambda_p )  \} \{ (F_{a_1} - F_{a_2})^\top r - \lambda_p \} + F_{a_2}^\top r + \lambda_p - F_{a_1}^\top p,
    \end{align*}
    which concludes the proof.
\end{proof}

We now make an important remark that is a consequence of \Cref{lemma:utility}. Before that, we need to introduce some additional quantities. First, given $\lambda\in\mathbb{R}$, we recall the definition of $f(\lambda)$ as the expected payment to $a_1$ of cheapest contract having an expected payment gap between the two actions of $\lambda$ as a solution to the following LP. More precisely, we define $f(\lambda)$ as follows
\begin{equation}
\label{eq:f-lambda}
f(\lambda)\coloneqq
\left\{
\begin{array}{ll}
\displaystyle \min_{p \in \mathcal{P}_1} & F_{a_1}^\top p \\[0.5em]
\text{s.t.} & (F_{a_1}-F_{a_2})^\top p = \lambda
\end{array}
\right.
\end{equation}
Then, for $\lambda \in \mathbb{R}$ and $\theta \in [0,1]$, we define the expected utility of the principal as a function of $\lambda$ as 
\begin{align*}
    & \widetilde U^P(\lambda, \theta) \coloneqq \mathds{1} \{ \theta_t \le \tilde \theta_{\max}(\lambda) \}[(F_{a_1} - F_{a_2})^\top r - \lambda] + F_{a_2}^\top r + \lambda - f(\lambda), 
\end{align*}
and we let 
\begin{align*}
    & \Lambda_1 \coloneqq \{ \lambda \in \Reals: \exists p \in \mathcal{P}_1 \land (F_{a_1} - F_{a_2})^\top p = \lambda  \}
\end{align*}
be the set of expected payment gaps that can be induced by contracts in $\mathcal{P}_1$.
Then, \Cref{lemma:utility} directly implies that, for any $\lambda \in \Lambda_1$:
\begin{align*}
    \min_{p \in \mathcal{P}_1: (F_{a_1}- F_{a_2})^\top p = \lambda} U^P(p, b(\theta,p)) = \widetilde U^P(\lambda, \theta).
\end{align*}

This simple observation is important to obtain the following key lemma. This lemma allows us to turn the original optimization problem in the $m$-dimensional space of contracts into a one-dimensional problem over $\lambda$.

\begin{lemma}[From $p$'s to $\lambda$'s]\label{lemma:opt-rewrite}
    Suppose that $F_{a_1}^\top r > F_{a_2}^\top r$. For any sequence of types $\{\theta_t \}_{t=1}^T$, we have that:
    \begin{align}\label{eq:rewrite}
    \opt = \max_{\lambda \in [0, (F_{a_1} - F_{a_2})^\top r]} \sum_{t=1}^T \widetilde U^P(\lambda, \theta_t).
    \end{align} 
\end{lemma}
\begin{proof}
    First, we observe that
    \begin{align*}
        \Lambda_1 & = \{ \lambda \in \Reals: \exists p \in \mathcal{P}_1 \land (F_{a_1} - F_{a_2})^\top p = \lambda  \} \\
        & = \{ \lambda \in \Reals: \exists p \in [0,1]^m \land F_{a_1}^\top (r-p)  \ge  F_{a_2}^\top (r-p) \land (F_{a_1} - F_{a_2})^\top p = \lambda  \} \\
        & = \{ \lambda \in \Reals: \exists p \in [0,1]^m \land (F_{a_1} - F_{a_2})^\top r  \ge  \lambda \land (F_{a_1} - F_{a_2})^\top p = \lambda  \} \\
        & = \left[ \sum_{\omega: F_{a_2,\omega} > F_{a_1,\omega}} F_{a_1,\omega} - F_{a_2,\omega} , (F_{a_1} - F_{a_2})^\top r\right].
    \end{align*}
    Then, we can rewrite the optimal value that can be extracted by the principal as follows:
    \begin{align*}
        \opt & = \sup_{p \in \mathcal{P}_1} \sum_{t=1}^T U^P(p, b(\theta_t, p)) \tag{\Cref{lemma:adv-opt-p1}} \\
        & = \sup_{p \in \mathcal{P}_1} \sum_{t=1}^T \widetilde U^P( (F_{a_1} - F_{a_2})^\top p, \theta_t) \\
        & = \sup_{\lambda \in \Lambda_1} \sum_{t=1}^T \widetilde U^P( \lambda, \theta_t),
    \end{align*}
    where the second equality holds by the remark above.
    At this point, consider $\lambda \in \Lambda_1$ such that $\lambda < 0$. Then, for any $\theta \in [0,1]$, it holds that:
    \begin{align*}
        \widetilde U ^P(\lambda, \theta) & = \mathds{1} \{ \theta \le \tilde \theta_{\max}(\lambda) \}[(F_{a_1} - F_{a_2})^\top r - \lambda] + F_{a_2}^\top r + \lambda - f(\lambda) \\
        & = F_{a_2}^\top r + \lambda - f(\lambda) \tag{$\tilde \theta_{\max}(\lambda) < 0$} \\
        & \le F_{a_2}^\top r \tag{$\lambda < 0$ and $f(\lambda) \ge 0$}.
    \end{align*}
    Furthermore, at $\bar \lambda = 0$, we have that:
    \begin{align*}
        \widetilde U^P(\bar \lambda, \theta) & = \mathds{1}\{\theta =  0)[(F_{a_1} - F_{a_2})^\top r - \bar \lambda \} + F_{a_2}^\top r + \bar \lambda - f(\bar \lambda) \\
        & = \mathds{1}\{\theta =  0)[(F_{a_1} - F_{a_2})^\top r  \} + F_{a_2}^\top r \tag{$f(\bar \lambda) = 0$ and $\bar \lambda = 0$} \\
        & \ge F_{a_2}^\top r.
    \end{align*}
    Thus, for all $\lambda < 0$ and any $\theta$, $\widetilde U^P(\lambda, \theta) \le \widetilde U^P(\bar \lambda, \theta)$.
    Therefore, we obtained that:
    \[
    \opt = \sup_{\lambda \in \Lambda_1} \sum_{t=1}^T \widetilde U^P( \lambda, \theta_t) = \sup_{\lambda \in [0, (F_{a_1} - F_{a_2})^\top r]} \sum_{t=1}^T \widetilde U^P(\lambda, \theta_t).
    \]
    Finally, we prove that the sup is attained.
    To this end, fix $t \in [T]$ and observe that $\lambda \to \mathds{1} \{ \theta \le \tilde \theta_{\max}(\lambda)\}$ is upper semicontinuous for any $\theta$. \footnote{See \Cref{app:others} for the definition of upper semicontinuity.} Furthermore, the product of two positive upper semicontinuous functions is upper semicontinuous, and hence $\mathds{1}\{\theta \le \tilde \theta_{\max}(\lambda)\} \left[(F_{a_1}-F_{a_2})^\top r - \lambda \right]$ is upper semicontinuous since $F_{a_1}^\top r > F_{a_2}^\top r$ and $\lambda \in [0, (F_{a_1} - F_{a_2})^\top r]$. Moreover, the sum of two upper semicontinuous functions is upper semicontinuous and thus $\sum_{t=1}^T \widetilde U^P(\lambda, \theta_t)$ is upper semicontinuous as well. Any upper semicontinuous function on a compact domain attains a maximum on that domain (\Cref{thm:opt-usc}), and therefore the sup is attained. This concludes the proof.
\end{proof}

Finally, we prove two important properties of the utility as a function of $\lambda$.
First, the next lemma shows that the variation of the one-dimensional objective $\widetilde U^P(\lambda,\theta)$ as $\lambda$ changes is controlled by the variation of the function $f(\lambda)$. As we discussed in \Cref{sec:ub-any-seq}, this property is crucial in the discretization argument.

\begin{lemma}[Objective Function Differences]\label{lemma:obj-diff}
    Consider $\lambda_2 \ge \lambda_1 \ge 0$ and any $\theta \in [0,1]$. It holds that:
    \begin{align*}
        \widetilde U^P(\lambda_1, \theta) - \widetilde U^P(\lambda_2,  \theta) \le f(\lambda_2) - f(\lambda_1). 
    \end{align*}
\end{lemma}
\begin{proof}
    The proof follows by using the definition of $\widetilde U^P(\lambda,\theta)$. Specifically:
\begin{align*}
    \widetilde U^P(\lambda_1, \theta) - \widetilde U^P(\lambda_2, \theta)
    = \lambda_1  &- \lambda_2 + f(\lambda_2) - f(\lambda_1) \\
    &+ \mathds{1}\{\theta \le \tilde \theta_{\max}(\lambda_1)\}
    \left[(F_{a_1}-F_{a_2})^\top r - \lambda_1 \right] \\
    &- \mathds{1}\{\theta \le \tilde \theta_{\max}(\lambda_2)\}
    \left[(F_{a_1}-F_{a_2})^\top r - \lambda_2 \right].
\end{align*}
Now, since $\lambda_1 \le \lambda_2$, we have that $\mathds{1}\{\theta \le \tilde \theta_{\max}(\lambda_1)\} \le \mathds{1}\{\theta \le \tilde \theta_{\max}(\lambda_2)\}$, thus:
\begin{align*}
    \widetilde U^P(\lambda_1, \theta) - \widetilde U^P(\lambda_2, \theta)
    &\le \lambda_1 - \lambda_2 + f(\lambda_2) - f(\lambda_1) \\
    &\quad + \mathds{1}\{\theta \le \tilde \theta_{\max}(\lambda_2)\}\left[\lambda_2 - \lambda_1 \right] \\
    &\le f(\lambda_2) - f(\lambda_1).
\end{align*}
which concludes the proof.
\end{proof}

Finally, we conclude by highlighting some properties of $f(\lambda)$ that we use in our discretization argument in the next section.

\begin{lemma}[Characterization of $f(\lambda)$]\label{lemma:f-lambda-properties}  
Consider $f(\lambda)$ and suppose $F_{a_1}^\top r > F_{a_2}^\top r$. Then, it holds that $f(0) = 0$. Moreover $f(\lambda) \le 1$ for all $\lambda$'s for which the problem is feasible. Finally, $f(\lambda)$ is well-defined, continuous and strictly increasing in the interval $[0, (F_{a_1} - F_{a_2})^\top r]$.
\end{lemma}
\begin{proof}
    At $0$, $f(0)$ reduces to the following optimization problem:
    \begin{align}\label{eq:char-proof-eq1}
        \min_{p \in \mathcal{P}_1} F_{a_1}^\top p \quad \text{~s.t.~} (F_{a_1}-F_{a_2})^\top p = 0. 
    \end{align}
    Observe that, since $p \ge 0$ and $F_{a_1}\in\Delta(\Omega)$, $f(\lambda)$ is always greater than $0$. Furthermore, since $F_{a_1}^\top r > F_{a_2}^\top r$, the null contract $\bar p = (0, \dots, 0)$ belongs to $\mathcal{P}_1$. Indeed:
    \[
    (F_{a_1} - F_{a_2})^\top \bar p = 0 < (F_{a_1} - F_{a_2})^\top r.
    \]
    Since the null contract belongs to the feasible region of \eqref{eq:char-proof-eq1} and since it yields $0$ as objective function, it is optimal.

    Next, suppose that the optimization problem that defines $f(\lambda)$ admits a non-empty feasible region. Then, for any $p \in \mathcal{P}_1$, $F_{a_1}^\top p \le 1$. Thus, $f(\lambda)$ is always bounded by $1$.

    Finally, let $\lambda \in [0, (F_{a_1} - F_{a_2})^\top r]$ and observe that the feasible region of \eqref{eq:f-lambda} is non-empty. This is guaranteed by the fact that $p \in \mathcal{P}_1$ if $p \in [0,1]^m$ and $(F_{a_1} - F_{a_2})^\top p \le (F_{a_1} - F_{a_2})^\top r$. Thus, for any $\lambda \in [0, (F_{a_1} - F_{a_2})^\top r]$, we can always find $p$ such that $(F_{a_1} - F_{a_2})^\top p = \lambda$. Next, we show that $f$ is continuous. To this end, we show that $f$ is convex. Consider $\lambda_1, \lambda_2 \in [0, (F_{a_1} - F_{a_2})^\top r]$ and let $p_1, p_2$ be optimal contracts for $\lambda_1$ and $\lambda_2$ respectively, \ie $f(\lambda_1) = F_{a_1}^\top p_1$ and $f(\lambda_2) = F_{a_1}^\top p_2$. Let $p_{\alpha} = \alpha p_1 + (1-\alpha)p_2$ for $\alpha \in [0,1]$. Then, by linearity, $p_\alpha$ is feasible point of \eqref{eq:f-lambda} with $\lambda = \alpha \lambda_1 + (1-\alpha)\lambda_2$. Therefore, it holds that:
    \[
    f(\lambda_\alpha) \le F_{a_1}^\top p_\alpha = \alpha F_{a_1}^\top p_1 + (1-\alpha)F_{a_1}^\top p_2 = \alpha f(\lambda_1) + (1-\alpha) f(\lambda_2),
    \]
    and, thus $f$ is convex. Concerning the strict monotonicity, instead, consider $\lambda_1 < \lambda_2$, and let $p_2$ be an optimal solution for $\lambda_2$, \ie $f(\lambda_2) = F_{a_1}^\top p_2$. Let $\alpha = \tfrac{\lambda_1}{\lambda_2}$ and let $p = \alpha \cdot p_2$. Then, observe that 
    \[
    (F_{a_1} - F_{a_2})^\top p = \frac{\lambda_1}{\lambda_2}(F_{a_1} - F_{a_2})^\top p_2 = \lambda_1.
    \]
    Furthermore, since $p_2 \in \mathcal{P}_1$ also $p \in \mathcal{P}_1$. Hence, $p$ is a feasible point of \eqref{eq:f-lambda} with $\lambda$ set to $\lambda_1$. Therefore, it holds that:
    \[
    f(\lambda_1) \le F_{a_1}^\top p = \frac{\lambda_1}{\lambda_2} F_{a_1}^\top p_2 = \frac{\lambda_1}{\lambda_2} f(\lambda_2) < f(\lambda_2),
    \]
    where in the last point, we have used that for $\lambda_2 > 0$, $f(\lambda_2)$ is strictly positive due to the constraint $(F_{a_1} - F_{a_2})^\top p = \lambda_2$. This concludes the proof.
\end{proof}

\subsection{Properties of the non-uniform discretization}\label{app:sec-disc}

We now prove some key properties of the non-uniform discretization that we introduced in \Cref{sec:ub-any-seq}.

\begin{lemma}[Discretization Error]\label{lemma:discretization}
    Let $\epsilon > 0$ and suppose that $F_{a_1}^\top r > F_{a_2}^\top r$. Then, for any sequence $\{ \theta_t \}_{t=1}^T$, the following holds:  
    \begin{align}
        & |\Lambda^{(\epsilon)}| \le \frac{1}{\epsilon} + 1 \label{eq:disc-count} \\
        & \opt - \max_{\lambda \in \Lambda^{(\epsilon)}} \sum_{t=1}^T \widetilde U^P(\lambda,\theta_t) \le T\epsilon. \label{eq:disc-err}
    \end{align}
\end{lemma}
\begin{proof}
    We first prove that $|\Lambda^{(\epsilon)}| \le \frac{1}{\epsilon} + 1$.
    By definition, it holds that:
    \[
        (|\Lambda^{(\epsilon)}|-1) \epsilon \le  \sum_{i = 1}^{|\Lambda^{(\epsilon)}|} \left( f(\lambda^{(i)}) - f(\lambda^{(i-1)}) \right) = f(\lambda^{(|\Lambda^{(\epsilon)}|)}) - f(\lambda^{(0)}) \le 1, 
    \]
    where in the last step we used $f(\lambda^{(0)}) = 0$ and $f(\cdot) \le 1$ (\Cref{lemma:f-lambda-properties}). Thus, \Cref{eq:disc-count} holds.

    Next, we continue by proving \Cref{eq:disc-err}. By applying  \Cref{lemma:opt-rewrite}, we have that there exists $\lambda^{\star} \in [0, (F_{a_1} - F_{a_2})^\top r]$ such that $\opt = \sum_{t=1}^T \widetilde U^P(\lambda^{\star}, \theta_t)$. Let $\tilde \lambda = \min \{ \lambda \in \Lambda^{(\epsilon)}: \lambda \ge \lambda^{\star} \}$. Observe that, by construction, $\tilde \lambda$ always exists as $\max_{\lambda \in \Lambda^{(\epsilon)}} \lambda = (F_{a_1} - F_{a_2})^\top r$ by definition. Now, suppose that $\tilde \lambda \ne \lambda^{\star}$ (otherwise \Cref{eq:disc-err} is direct) and observe that the set $\overline \Lambda^{(\epsilon)} = \{\lambda \in \Lambda^{(\epsilon)}: \lambda \le \lambda^{\star} \}$ is always non-empty, since $0 \in \Lambda^{(\epsilon)}$ for all $\epsilon > 0$ and $\lambda^{\star} \ge 0$. Denote by $\overline \lambda$ the following quantity $\overline \lambda = \max_{\lambda \in \overline \Lambda^{(\epsilon)}} \lambda$ and observe that:
    \begin{align}\label{eq:disc-err-helper-1}
        f(\tilde \lambda) - f(\overline \lambda) \le \epsilon. 
    \end{align}
    The proof of \Cref{eq:disc-err-helper-1} is direct once we realize that, by definition, there exists $n \in \Naturals$ such that $\overline \lambda = \lambda^{(n)}$ and $\tilde \lambda = \lambda^{(n+1)}$ (\ie $\overline \lambda $ and $\tilde \lambda$ are adjacent elements of the sequence we constructed).
    
    Then, we have that:
    \begin{align*}
        \opt - \max_{\lambda \in \Lambda^{(\epsilon)}} \sum_{t=1}^T \widetilde U^P(\lambda, \theta_t) & \le \sum_{t=1}^T \widetilde U^P(\lambda^{\star}, \theta_t) - \widetilde U^{P}(\tilde \lambda, \theta_t)  \\
        & \le T\left( f(\tilde \lambda) - f(\lambda^{\star}) \right) \tag{\Cref{lemma:obj-diff} and $\tilde \lambda \ge \lambda^{\star}$} \\
        & \le T\left( f(\tilde \lambda) - f(\overline \lambda) \right) \tag{$\overline \lambda \le \lambda^{\star}$ and $\lambda \to f(\lambda)$ is increasing}\\
        & \le T\epsilon, \tag{\Cref{eq:disc-err-helper-1}}
    \end{align*}
    This concludes the proof.
\end{proof}

\subsection{Combining the results}\label{app:regret-anal-thm1}
Finally, we can combine these results and prove \Cref{thm:ub1}.

\begin{proof}[Proof of \Cref{thm:ub1}]
    If $F_{a_2}^\top r \ge F_{a_1}^\top r$, the result is direct from \Cref{lemma:easy-case}.

    Next, consider $F_{a_1}^\top r > F_{a_2}^\top r$. Let $\epsilon > 0$ and consider the following set of contracts:
    \begin{align*}
        \mathcal{P}^{(\epsilon)} = \{p \in [0,1]^m : \exists \lambda  \in \Lambda^{(\epsilon)} \text{~and~} F_{a_1}^\top p = f(\lambda)  \}.
    \end{align*}
    Consider any minimax optimal algorithm for finite-armed adversarial bandits and run it using as arm space $\mathcal{P}^{(\epsilon)}$, \ie \citet{audibert2010regret}. Recall that the algorithm in \citet{audibert2010regret} achieves $\mathcal{O}(\sqrt{TK})$ regret on adversarial oblivious bandits with $K$ arms. 

    Then, it holds that:
    \begin{align*}
        R_T & =  \opt - \mathbb{E}\left[\sum_{t=1}^T U^P(p_t, b(\theta_t, p_t) \right] \\
        & =  \opt -  \max_{p \in \mathcal{P}^{(\epsilon)}} \sum_{t=1}^T U^P(p, b(\theta_t, p)) + \max_{p \in \mathcal{P}^{(\epsilon)}} \sum_{t=1}^T U^P(p, b(\theta_t, p)) - \mathbb{E}\left[\sum_{t=1}^T U^P(p_t, b(\theta_t, p_t) \right] \\
        & = \opt -  \max_{\lambda \in {\Lambda}^{(\epsilon)}} \sum_{t=1}^T \widetilde U^P(\lambda, \theta_t) + \max_{p \in \mathcal{P}^{(\epsilon)}} \sum_{t=1}^T U^P(p, b(\theta_t, p)) - \mathbb{E}\left[\sum_{t=1}^T U^P(p_t, b(\theta_t, p_t) \right] \\ 
        & \le T \epsilon + \max_{p \in \mathcal{P}^{(\epsilon)}} \sum_{t=1}^T U^P(p, b(\theta_t, p)) - \mathbb{E}\left[\sum_{t=1}^T U^P(p_t, b(\theta_t, p_t) \right] \tag{\Cref{lemma:discretization}} \\
        & \in \mathcal{O}\left( T \epsilon + \sqrt{|\mathcal{P}^{(\epsilon)}|T} \right) \tag{Minimax optimality} \\
        & = \mathcal{O}\left( T \epsilon + \sqrt{\frac{T}{\epsilon}} \right) \tag{\Cref{lemma:discretization}}.
    \end{align*}
    where the second inequality follows from \Cref{lemma:utility} together with the definition of $f(\lambda)$ and the fact that $p \in \mathcal{P}^{(\epsilon)}$ if and only if there exists $\lambda \in \Lambda^{(\epsilon)}$ such that $F_{a_1}^\top p = f(\lambda)$.
    
    Now, choosing $\epsilon = T^{-1/3}$ yields the result.
\end{proof}

\newpage
\section{Proof of \Cref{thm:lb1}}\label{sec:lb}

In this section, we give a formal proof of the $\Omega(T^{2/3})$ lower bound, which holds even when the type sequence is stochastic and comes from an absolutely continuous distribution. 

\subsection{Construction}
\begin{figure}[!tp]
    \centering
    \scalebox{0.7}{\pgfplotsset{compat=1.18}
\definecolor{plotblue}{HTML}{0064E0}

\pgfmathsetmacro{\tval}{0.15}   
\pgfmathsetmacro{\cval}{0.48}  
\pgfmathsetmacro{\dval}{0.1}  

\pgfmathsetmacro{\cd}{\cval - \dval}       
\pgfmathsetmacro{\cpd}{\cval + \dval}      
\pgfmathsetmacro{\omt}{1 - \tval}          

\pgfmathsetmacro{\yFlat}{1 / (1 - \tval)}
\pgfmathsetmacro{\ySpike}{\tval * 2 /
    ((1 - \cval - \dval) * (1 - \cval + \dval) )}
\pgfmathsetmacro{\yRight}{1 / \tval}          
\pgfmathsetmacro{\ymx}{max(\yFlat, max(\ySpike, \yRight)) * 1.14}

\pgfmathsetmacro{\jumpT}{\tval / (1 - \tval)^2}     
\pgfmathsetmacro{\jumpCD}{\tval / (1 - \cd)^2}      
\pgfmathsetmacro{\jumpCPD}{\tval / (1 - \cpd)^2}    

\begin{tikzpicture}
  \begin{axis}[
    width  = 13cm,
    height = 5.0cm,
    xmin   = -0.00,  xmax = 1.03,
    ymin   = -0.04,  ymax = \ymx,
    axis x line = bottom,
    axis y line = left,
    xlabel = {$x$},
    ylabel = {$\gamma_c(x)$},
    xlabel style     = {below, font=\small},
    ylabel style     = {above, font=\small},
    tick label style = {font=\small},
    xtick       = {\tval, \cd, \cval, \cpd, \omt, 1},
    xticklabels = {$t$, $c{-}\delta$, $\phantom{\delta}c\phantom{\delta}$, $c{+}\delta$, $1{-}t$, $1$},
    extra x ticks = {0},
    extra x tick labels = {$0$},
    ytick = \empty,
    clip  = true,
    clip mode = individual,
  ]



  \addplot[fill=plotblue, fill opacity=0.0, draw=none]
    coordinates {(0,0) (0,\yFlat) (\tval,\yFlat) (\tval,0)};

  \addplot[fill=plotblue, fill opacity=0.0, draw=none,
           domain=\tval:\cd, samples=80]
    {\tval / (1-x)^2} \closedcycle;

  \addplot[fill=plotblue, fill opacity=0.0, draw=none]
    coordinates {(\cd,0) (\cd,\ySpike) (\cval,\ySpike) (\cval,0)};

  \addplot[fill=plotblue, fill opacity=0.0, draw=none,
           domain=\cpd:\omt, samples=100]
    {\tval / (1-x)^2} \closedcycle;


  \addplot[color=plotblue, line width=2.6pt]
    coordinates {(0, \yFlat) (\tval, \yFlat)};

  \addplot[color=plotblue, line width=2.6pt,
           domain=\tval:\cd, samples=80]
    {\tval / (1-x)^2};

  \addplot[color=plotblue, line width=2.6pt]
    coordinates {(\cd, \ySpike) (\cval, \ySpike)};

  \addplot[color=plotblue, line width=2.6pt]
    coordinates {(\cval, 0) (\cpd, 0)};

  \addplot[color=plotblue, line width=2.6pt,
           domain=\cpd:\omt, samples=100]
    {\tval / (1-x)^2};

  \addplot[color=plotblue, line width=2.6pt]
    coordinates {(\omt, 0) (1, 0)};


  \draw[dashed, color=plotblue, line width=0.5pt]
    (axis cs:\tval, 0) -- (axis cs:\tval, \yFlat);
     
%
  \draw[dashed, color=plotblue, line width=0.5pt]
    (axis cs:\cd, 0) -- (axis cs:\cd, \ySpike);
%
  \draw[dashed, color=plotblue, line width=0.5pt]
    (axis cs:\cval, 0) -- (axis cs:\cval, \ySpike);
%
  \draw[dashed, color=plotblue, line width=0.5pt]
    (axis cs:\cpd, 0) -- (axis cs:\cpd, \jumpCPD);
%
  \draw[dashed, color=plotblue, line width=0.5pt]
   (axis cs:\omt, 0) -- (axis cs:\omt, \yRight);

  \end{axis}
\end{tikzpicture}}
    \caption{Type distribution density of the lower bound instance.}
    \label{fig:pdf}
\end{figure}
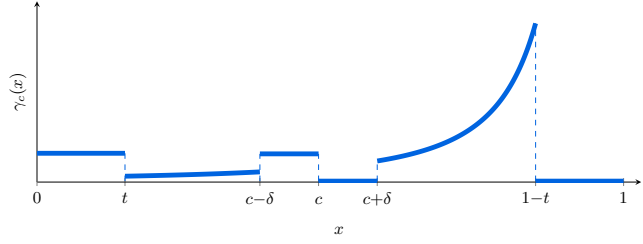
We consider the following $K+1$ instances $\{\I_0,\I_1,\ldots, \I_K\}$, where $\I_0$ will be the \emph{base} instance and $\{\I_1,\ldots, \I_K\}$ will be the \emph{perturbed} instances. Let $\{\nu_0,\nu_1,\ldots, \nu_K\}$ be the joint probability distributions over the $T$ turns of the interaction between the instances and the (possibly randomized) algorithm.
Each perturbed instance will be parametrized by two parameters $(c,\delta)$, which we will fix in the following.
We now describe a generic instance parametrized by $(c,\delta)\in[0,1]^2$. We have $\Gamma_{c,\delta}$, the type' distribution (that we will describe shortly), $2$ outcomes $\omega_1$ and $\omega_2$ such that $r_{\omega_1}=1$ and $r_{\omega_2}=0$, the cost of performing action $a_1$ is $c_{a_1}=1$ and the cost of performing action $a_2$ is $c_{a_2}=0$, action $a_1$ leads deterministically to $\omega_1$ and action $a_2$ leads deterministically to $\omega_2$.

Fix $t=1/5$ and any $\epsilon>0$. To each instance $\I_k$ with $k\ge 1$, we are going to associate a type distribution $\Gamma_k=\Gamma_{c_k,\delta_k}$, where $c_k=t+\bar\delta(2k-1)$, $\bar\delta:=\epsilon\frac{1-t}{t+2\epsilon}$, and finally, $\delta_k=\epsilon\frac{1-c_k}{t+\epsilon}$. The base instance $\I_0$ will be defined with $\Gamma_0=\Gamma_{0,0}$, whose definition will be clear once we give a concrete parametrized definition of $\Gamma_{c,\delta}$.

First, we characterize the shape of the principal utility for a generic instance parametrized by any $\Gamma$.
\begin{lemma}\label{lem:LB1}
    For each contract $p\in[0,1]^2$, and any type distribution $\Gamma$, the principal' utility is:
    \[
    U^P(p)=\mathbb{P}_{\Gamma}(\theta\le p_{\omega_1}-p_{\omega_2})(1-(p_{\omega_1}-p_{\omega_2}))-p_{\omega_2}.
    \]
\end{lemma}
\begin{proof}
    The agent will play action $a_1$ if and only if $p_{\omega_1}-\theta\cdot c_{a_1}=p_{\omega_1}-\theta\ge p_{\omega_2}$ so only if $\theta\le p_{\omega_1}-p_{\omega_2}$. Thus, the utility of the principal is
    \[
    \mathbb{P}_{\Gamma}(\theta\le p_{w_1}-p_{w_2})(1-p_{w_1})+\mathbb{P}_{\Gamma}(\theta> p_{w_1}-p_{w_2})(-p_{w_2})=\mathbb{P}_{\Gamma}(\theta\le p_{w_1}-p_{w_2})(1-p_{w_1}+p_{w_2})-p_{w_2},
    \]
    as claimed.
\end{proof}

\begin{remark}\label{rem:LB}
From this lemma, it is clear that we can focus on algorithms that play $p_{w_2}=0$. Indeed, for every algorithm that plays $(p_{\omega_1},p_{\omega_2})$, we can instantiate an algorithm that plays $(p_{\omega_1}-p_{\omega_2},0)$ and the feedback under any instance $\{\I_k\}_{k=0}^K$, has the same distribution and larger utility.
\end{remark}

Thus, for every instance $\{\I_0,\ldots, \I_K\}$ we can define 
\[
U^P_k(x):=\mathbb{P}_{\Gamma_k}(\theta\le x)(1-x),
\]
as the principal expected utility of proposing contract $p=(x,0)$ to the agent.

Now we can define the type distribution $\Gamma_{c,\delta}$ as follows. The type distribution is given by its probability density $\gamma_{c,\delta}$.
\begin{align}\label{eq:density}
    \gamma_{c,\delta}(x)=\begin{cases}
        \frac{1}{1-t} &\text{if }x\in[0,t)\\
        \frac{t}{(1-x)^2} &\text{if }x\in[t,c-\delta)\\
        \frac{2t}{(1-c-\delta)(1-c+\delta)} &\text{if }x\in[c-\delta,c)\\
        0 &\text{if }x\in[c,c+\delta)\\
        \frac{t}{(1-x)^2} &\text{if }x\in[c+\delta,1-t)\\
        0 &\text{if }x\in[1-t,1]\\
    \end{cases}
\end{align}

A visual representation of the density is shown in \Cref{fig:pdf}.
The pdf has two constant regions: $[0,t)$ and $[c-\delta,c)$, while on $[t,c-\delta)$ and $[c+\delta,1-t)$ it is a power-law. Everywhere else it is $0$.
From \Cref{lem:LB1} and \Cref{rem:LB}, we know that the principal utility is the function $x\mapsto (1-x)\cdot\int_{0}^x\gamma_{c,\delta}(s)ds$.
%
%
%
%
We now characterize the main features of the instances, which are 1) the reward structure and 2) the information structure. An illustration can also be found in \Cref{fig:figures}.

\begin{lemma}\label{lem:charLB}
    For each $k\in[K]$, for a contract $p=(x,0)$, the principal expected utility $U^P(p)$ in instance $\I_k$ and the feedback distribution is a Bernoulli with probability $\mathbb{P}_{\Gamma_k}(\theta\le x)$ are as follows:
    \[
\mathbb{P}_{\Gamma_k}(\theta\le x)=
\begin{cases}
 \frac{x}{1-t} & \text{if } x\in[0,t),\\
\frac{t}{1-x} & \text{if } x\in[t,c_k-\delta_k),\\
\frac{t}{1-c_k+\delta_k}+(x-c_k+\delta_k)\frac{2t}{(1-c_k-\delta_k)(1-c_k+\delta_k)}
& \text{if } x\in[c_k-\delta_k,c_k),\\
\frac{t}{1-c_k-\delta_k} & \text{if } x\in[c_k,c_k+\delta_k),\\
\frac{t}{1-x} & \text{if } x\in[c_k+\delta_k,1-t),\\
1 & \text{if } x\in[1-t,1].
\end{cases}
\]
Moreover, the base instance is characterized by:
    \[
\mathbb{P}_{\Gamma_0}(\theta\le x)=
\begin{cases}
 \frac{x}{1-t} & \text{if } x\in[0,t),\\
\frac{t}{1-x} & \text{if } x\in[t,1-t),\\
1 & \text{if } x\in[1-t,1].
\end{cases}
\]
\end{lemma}
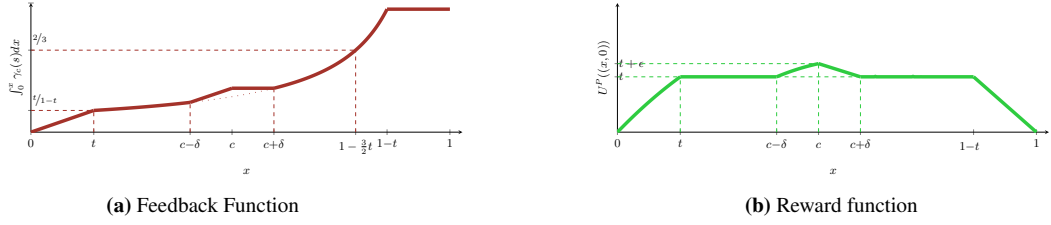
\begin{figure}[tp!]
\centering
\hspace{-1cm}
\begin{subfigure}[b]{0.47\textwidth}
\centering
    \scalebox{0.5}{\pgfplotsset{compat=1.18}
\definecolor{plotred}{HTML}{A4322A}

\pgfmathsetmacro{\tval}{0.15}   
\pgfmathsetmacro{\cval}{0.48}  
\pgfmathsetmacro{\dval}{0.1}  

\pgfmathsetmacro{\cd}{\cval - \dval}       
\pgfmathsetmacro{\cpd}{\cval + \dval}      
\pgfmathsetmacro{\omt}{1 - \tval}          

\pgfmathsetmacro{\yFlat}{1 / (1 - \tval)}
\pgfmathsetmacro{\ySpike}{\tval * 2 /
    ((1 - \cval - \dval) * (1 - \cval + \dval) )}
\pgfmathsetmacro{\yRight}{1 / \tval}          
\pgfmathsetmacro{\ymx}{max(\yFlat, max(\ySpike, \yRight)) * 1.14}

\pgfmathsetmacro{\jumpT}{\tval / (1 - \tval)^2}     
\pgfmathsetmacro{\jumpCD}{\tval / (1 - \cd)^2}      
\pgfmathsetmacro{\jumpCPD}{\tval / (1 - \cpd)^2}    

\begin{tikzpicture}
  \begin{axis}[
    width  = 13cm,
    height = 5.0cm,
    xmin   = -0.00,  xmax = 1.03,
    ymin   = 0.0,  ymax = 1.05,
    axis x line = bottom,
    axis y line = left,
    xlabel = {$x$},
    ylabel = {$\int_0^x\gamma_c(s)dx$},
    xlabel style     = {below, font=\small},
    ylabel style     = {above, font=\small},
    tick label style = {font=\small},
    xtick       = {\tval, \cd, \cval, \cpd,  1-3*\tval/2, \omt, 1},
    xticklabels = {$t$, $c{-}\delta$, $\phantom{\delta}c\phantom{\delta}$, $c{+}\delta$, $1-\tfrac32t$, $1{-}t$, $1$},
    yticklabels = {$\nicefrac t{1-t}$, $\nicefrac23$},
    y tick label style={font=\small, rotate=0, right, anchor=south west},
    extra x ticks = {0},
    extra x tick labels = {$0$},
    ytick = {\tval/(1-\tval), 2/3, 1.05},
    clip  = true,
    clip mode = individual,
  ]



  \addplot[fill=plotred, fill opacity=0.0, draw=none, domain=0:\tval, samples=80]
    {x / (1-\tval)} \closedcycle;

  \addplot[fill=plotred, fill opacity=0.0, draw=none,
           domain=\tval:\cd, samples=80]
    {\tval / (1-x)} \closedcycle;

  \addplot[fill=plotred, fill opacity=0.0, draw=none, domain=\cd:\cval, samples=80]
    {\tval / (1-\cd)+(x-\cd)*2*\tval/((1-\cpd)*(1-\cd))} \closedcycle;

  \addplot[fill=plotred, fill opacity=0.0, draw=none, domain=\cval:\cpd, samples=80]
    {\tval / (1-\cpd)} \closedcycle;

  \addplot[fill=plotred, fill opacity=0.0, draw=none,
           domain=\cpd:\omt, samples=100]
    {\tval / (1-x)} \closedcycle;

  \addplot[fill=plotred, fill opacity=0.0, draw=none,
           domain=\omt:1, samples=100]
    {1} \closedcycle;


  \addplot[color=plotred, line width=2.6pt, domain=0:\tval, samples=80]
    {x / (1-\tval)};

  \addplot[color=plotred, line width=2.6pt,
           domain=\tval:\cd, samples=80]
    {\tval / (1-x)};

  \addplot[color=plotred, line width=2.6pt, domain=\cd:\cval, samples=80]
    {\tval / (1-\cd)+(x-\cd)*2*\tval/((1-\cpd)*(1-\cd))};

  \addplot[color=plotred, line width=2.6pt, domain=\cval:\cpd, samples=80]
    {\tval / (1-\cpd)};

  \addplot[color=plotred, line width=2.6pt,
           domain=\cpd:\omt, samples=100]
    {\tval / (1-x)};

  \addplot[color=plotred, line width=2.6pt]
    coordinates {(\omt, 1) (1, 1)};

  \addplot[loosely dotted, color=plotred, line width=.5pt,
           domain=\cd:\cpd, samples=100]
    {\tval / (1-x)};


%
  \draw[dashed, color=plotred, line width=.5pt]
    (axis cs:\cd, 0) -- (axis cs:\cd, {\tval/(1-\cd)});
%
%
  \draw[dashed, color=plotred, line width=0.5pt]
    (axis cs:\cpd, 0) -- (axis cs:\cpd, {\tval/(1-\cpd)});

  \draw[dashed, color=plotred, line width=0.5pt]
    (axis cs:0, {\tval / (1-\tval)}) -- (axis cs:{\tval}, {\tval / (1-\tval)});
  \draw[dashed, color=plotred, line width=0.5pt]
    (axis cs: {\tval}, 0) -- (axis cs:{\tval}, {\tval / (1-\tval)});

  \draw[dashed, color=plotred, line width=0.5pt]
    (axis cs:0, {2/3}) -- (axis cs:{1-3*\tval/2}, {2/3});
  \draw[dashed, color=plotred, line width=0.5pt]
    (axis cs: {1-3*\tval/2}, 0) -- (axis cs:{1-3*\tval/2}, {2/3});
%

  \end{axis}
\end{tikzpicture}}
    \caption{Feedback Function}
    \label{fig:cdf}
\end{subfigure}
\hfill
\begin{subfigure}[b]{0.47\textwidth}
\centering
\hspace{-1cm}
    \scalebox{0.5}{\pgfplotsset{compat=1.18}
\definecolor{plotgreen}{HTML}{2ECC40}

\pgfmathsetmacro{\tval}{0.15}   
\pgfmathsetmacro{\cval}{0.48}  
\pgfmathsetmacro{\dval}{0.1}  

\pgfmathsetmacro{\cd}{\cval - \dval}       
\pgfmathsetmacro{\cpd}{\cval + \dval}      
\pgfmathsetmacro{\omt}{1 - \tval}          

\pgfmathsetmacro{\yFlat}{1 / (1 - \tval)}
\pgfmathsetmacro{\ySpike}{\tval * 2 /
    ((1 - \cval - \dval) * (1 - \cval + \dval) )}
\pgfmathsetmacro{\yRight}{1 / \tval}          
\pgfmathsetmacro{\ymx}{max(\yFlat, max(\ySpike, \yRight)) * 1.14}

\pgfmathsetmacro{\jumpT}{\tval / (1 - \tval)^2}     
\pgfmathsetmacro{\jumpCD}{\tval / (1 - \cd)^2}      
\pgfmathsetmacro{\jumpCPD}{\tval / (1 - \cpd)^2}    

\begin{tikzpicture}
  \begin{axis}[
    width  = 13cm,
    height = 5.0cm,
    xmin   = -0.00,  xmax = 1.03,
    ymin   = 0.0,  ymax = \tval+0.2,
    axis x line = bottom,
    axis y line = left,
    xlabel = {$x$},
    ylabel = {$U^P((x,0))$},
    xlabel style     = {below, font=\small},
    ylabel style     = {above, font=\small},
    tick label style = {font=\small},
    xtick       = {\tval, \cd, \cval, \cpd, \omt, 1},
    ytick       = {\tval, \tval*(1-\cval) / (1-\cpd)},
    xticklabels = {$t$, $c{-}\delta$, $\phantom{\delta}c\phantom{\delta}$, $c{+}\delta$, $1{-}t\phantom{\tfrac32}$, $1$},
    extra x ticks = {0},
    extra x tick labels = {$0$},
    yticklabels = {$t$,$t+\epsilon$},
    y tick label style={font=\small, rotate=0, right, anchor=west},
    clip  = true,
    clip mode = individual,
  ]



  \addplot[fill=plotgreen, fill opacity=0.0, draw=none, domain=0:\tval, samples=80]
    {x*(1-x) / (1-\tval)} \closedcycle;

  \addplot[fill=plotgreen, fill opacity=0.0, draw=none,
           domain=\tval:\cd, samples=80]
    {\tval*(1-x) / (1-x)} \closedcycle;

  \addplot[fill=plotgreen, fill opacity=0.0, draw=none, domain=\cd:\cval, samples=80]
    {\tval*(1-x) / (1-\cd)+(1-x)*(x-\cd)*2*\tval/((1-\cpd)*(1-\cd))} \closedcycle;

  \addplot[fill=plotgreen, fill opacity=0.0, draw=none, domain=\cval:\cpd, samples=80]
    {\tval*(1-x) / (1-\cpd)} \closedcycle;

  \addplot[fill=plotgreen, fill opacity=0.0, draw=none,
           domain=\cpd:\omt, samples=100]
    {\tval*(1-x) / (1-x)} \closedcycle;

  \addplot[fill=plotgreen, fill opacity=0.0, draw=none,
           domain=\omt:1, samples=100]
    {1*(1-x)} \closedcycle;


  \addplot[color=plotgreen, line width=2.6pt, domain=0:\tval, samples=80]
    {x*(1-x) / (1-\tval)};

  \addplot[color=plotgreen, line width=2.6pt,
           domain=\tval:\cd, samples=80]
    {\tval*(1-x) / (1-x)};

  \addplot[color=plotgreen, line width=2.6pt, domain=\cd:\cval, samples=80]
    {\tval*(1-x) / (1-\cd)+(1-x)*(x-\cd)*2*\tval/((1-\cpd)*(1-\cd))};

  \addplot[color=plotgreen, line width=2.6pt, domain=\cval:\cpd, samples=80]
    {\tval*(1-x) / (1-\cpd)};

  \addplot[color=plotgreen, line width=2.6pt,
           domain=\cpd:\omt, samples=100]
    {\tval*(1-x) / (1-x)};

  \addplot[color=plotgreen, line width=2.6pt,
           domain=\omt:1, samples=100]
    {(1-x)};

  \addplot[dashed, color=plotgreen, line width=.5pt,
           domain=\cd:\cpd, samples=100]
    {\tval*(1-x) / (1-x)};


%
  \draw[dashed, color=plotgreen, line width=.5pt]
    (axis cs:\cd, 0) -- (axis cs:\cd, {\tval});
%
%
  \draw[dashed, color=plotgreen, line width=0.5pt]
    (axis cs:\cpd, 0) -- (axis cs:\cpd, {\tval});
%
  \draw[dashed, color=plotgreen, line width=0.5pt]
    (axis cs:0, {\tval }) -- (axis cs:{\tval}, {\tval });
  \draw[dashed, color=plotgreen, line width=0.5pt]
    (axis cs: {\tval}, 0) -- (axis cs:{\tval}, {\tval });
  \draw[dashed, color=plotgreen, line width=0.5pt]
    (axis cs: {\cval}, 0) -- (axis cs:{\cval}, {\tval*(1-\cval) / (1-\cpd)});
  \draw[dashed, color=plotgreen, line width=0.5pt]
    (axis cs:0, {{\tval*(1-\cval) / (1-\cpd)}}) -- (axis cs:{\cval}, {\tval*(1-\cval) / (1-\cpd) });

  \end{axis}
\end{tikzpicture}}
    \caption{Reward function}
    \label{fig:util}
\end{subfigure}
        
\caption{Feedback distribution and reward function of the lower bound instances.}
\label{fig:figures}
\end{figure}

We now observe some important features of our instances. For each instance $\I_k$ for $k\ge 1$, we define the interval $J_k:=[c_k-\bar\delta,c_k+\bar \delta]$ and $B_k:=[c_k,c_k+\delta_k/2]$. In \Cref{fig:figures} we have an illustration of the reward and feedback distribution.

\begin{lemma}\label{lem:LBtmp1}
    For any $K\le \frac{1}{20\epsilon}$ and $\epsilon<1/20$, we have
    \begin{enumerate}[label=(\roman*)]
        \item $\mathbb{P}_{\Gamma_k}(\theta\le x)=\mathbb{P}_{\Gamma_{k'}}(\theta\le x)$ for all $x\not\in J_{k}\cup J_{k'}$;
        \item $\delta_k\le \bar\delta$ for all $k\in[K]$;
        \item $J_k\subseteq[t, 1-\frac32t]$ for all $k\in[K]$;
        \item $U_k^P(x)\ge t+\epsilon/2$ for all $x\in B_k:=[c_k,c_k+\delta_k/2]$;
        \item $|\mathbb{P}_{\Gamma_k}(\theta\le x)-\mathbb{P}_{\Gamma_0}(\theta\le x)|\le 4\epsilon$ for all $x\in J_k$.
    \end{enumerate}
\end{lemma}
\begin{proof}
    For $(i)$ it is enough to observe that all type distributions coincide with the base instance outside of the intervals $J_k$.

    For $(ii)$ we can observe that $\delta_k$ is decreasing in $c_k$ and thus it is enough to prove the statement for $k=1$, and straightforward calculations show that $\bar\delta=\delta_1$.

    For $(iii)$ it is enough to observe that $c_k-\bar\delta=t+2\bar\delta(k-1)\ge t$ and thus $J_k\subseteq [t,\infty)$. Moreover, $c_k+\bar\delta\le c_K+\bar\delta=t+2K\bar\delta$. Since $t=1/5$, $\bar\delta=\frac{4\epsilon}{1+10\epsilon}$, we get $2K\bar\delta\le 2/5$ and thus $c_k+\bar\delta\le 1-\frac32t$, proving that $J_k\subseteq [t,1-\frac32t]$.

    For $(iv)$ it is sufficient to observe that $U_k^P(x)$ is decreasing in $x$ in the interval $B_k$, and that $U_k^P(c_k+\delta_k/2)=t+\epsilon/2$.

    For $(v)$ we can observe that by $(ii)$ and 
    \Cref{lem:charLB}, for all $x\notin[c_k-\delta_k,c_k+\delta_k]$
    \[
    \mathbb P_{\Gamma_k}(\theta\le x)=\mathbb P_{\Gamma_0}(\theta\le x).
    \]
    Thus, it suffices to consider $x\in[c_k-\delta_k,c_k+\delta_k]$. On this interval, the difference
    \[
        D(x):=\mathbb P_{\Gamma_k}(\theta\le x)-\mathbb P_{\Gamma_0}(\theta\le x)
    \]
    is nonnegative, increasing on $[c_k-\delta_k,c_k)$, and decreasing on $[c_k,c_k+\delta_k]$, hence it is maximized at $x=c_k$. Therefore
    \[
    \left|\mathbb P_{\Gamma_k}(\theta\le x)-\mathbb P_{\Gamma_0}(\theta\le x)\right| \le \frac{t}{1-c_k-\delta_k}-\frac{t}{1-c_k} = \frac{\epsilon}{1-c_k}.
    \]
    By item 3, $1-c_k\ge 3t/2$, so the right-hand side is at most
    \(
    \frac{2}{3t}\epsilon\le 4\epsilon,
    \)
    concluding the proof.
\end{proof}

%

A standard change of measure argument concluded the proof.
\begin{theorem}
For any (possibly randomized) learning algorithm, there exists $k\in[K]$ such that
    \[
    \opt-\mathbb{E}_{\nu_k}\left[\sum_{t=1}^TU^P(p_t)\right]=\Omega(T^{2/3}).
    \]
\end{theorem}
\begin{proof}
    Because of \Cref{rem:LB}, we can focus on algorithms that play on the single-dimensional line $x_t\in[0,1]$, and the associated contract is $p_t=(x_t,0)$.
    For any instance $k\in[K]$ let $N_k(T)=\sum_{t\in[T]}\mathbb{I}(x_t\in J_k)$ and $M_k(T)=\sum_{t\in[T]}\mathbb{I}(x_t\in B_k)$.
    \begin{claim}\label{claim:1}
        \[
        \textup{KL}(\nu_0,\nu_k)\le 10^4\epsilon^2\mathbb{E}_{\nu_0}[N_k(T)].
        \]
    \end{claim}
    \begin{proof}
    We can use the $\textup{KL}$ decomposition for continuous action spaces \citep[Lemma~15.8~(a)]{lattimore2020learning}. 
        \begin{align*}
            \textup{KL}(\nu_0,\nu_k)=\int_0^1\textup{kl}(\mathbb{P}_{\Gamma_k}(\theta\le x),\mathbb{P}_{\Gamma_0}(\theta\le x)) dG_{\nu_0}(x),
        \end{align*}
        where $G_{\nu_0}(B)=\mathbb{E}_{\nu}[\sum_{t\in[T]}\mathbb{I}(x_t\in B)]$ and $\textup{kl}(a,b)=a\log(\nicefrac ab)+(1-a)\log(\nicefrac{1-a}{1-b})$.
        By \Cref{lem:LBtmp1}, we have that
        \begin{align*}
            \textup{KL}(\nu_0,\nu_k)&=\int_{J_k}\textup{kl}(\mathbb{P}_{\Gamma_k}(\theta\le x),\mathbb{P}_{\Gamma_0}(\theta\le x))dG_{\nu_0}(x)\\
            &\le G_{\nu_0}(J_k) \sup_{x\in J_k}\textup{kl}(\mathbb{P}_{\Gamma_k}(\theta\le x),\mathbb{P}_{\Gamma_0}(\theta\le x))\\
            &=\mathbb{E}_{\nu_0}[N_k(T)]\textup{kl}(\mathbb{P}_{\Gamma_k}(\theta\le x),\mathbb{P}_{\Gamma_0}(\theta\le x)).
        \end{align*}
        Now we upper bound $\textup{kl}(a,b)\le\frac{(a-b)^2}{2\ell(1-u)}$, which holds for all $a,b\in[\ell,u]$. By using \Cref{lem:LBtmp1} we know that $\mathbb{P}_{\Gamma_0}(\theta\le x)\in[\nicefrac14,\nicefrac23]$, and from \Cref{lem:charLB} $(v)$ that $\mathbb{P}_{\Gamma_k}(\theta\le x)\in[\nicefrac14-4\epsilon,\nicefrac23+4\epsilon]$, which, for $\epsilon\le1/20$, is contained in $[1/20,13/15]=[\ell,u]$. This lets us conclude that 
        \[
        \textup{KL}(\nu_0,\nu_k)\le 10^4\epsilon^2\mathbb{E}_{\nu_0}[N_k(T)],
        \]
        concluding the proof.
    \end{proof}
    Moreover, the regret can be easily upper-bounded as follows:
    \begin{align}
        \opt-\mathbb{E}_{\nu_k}\left[\sum_{t=1}^TU^P(p_t)\right]&=T \left(t+\epsilon\right)-\mathbb{E}_{\nu_k}\left[\sum_{t\in[T]:x_t\in B_k(T)}^TU^P(p_t)+\sum_{t\in[T]:x_t\notin B_k(T)}^TU^P(p_t)\right]\notag\\
        &\ge T(t+\epsilon)-\mathbb{E}_{\nu_k}\left[M_k(T)\right](t+\epsilon)-(T-\mathbb{E}_{\nu_k}\left[M_k(T)\right])(t+\frac\epsilon2)   \notag \\
        &=\frac\epsilon2(T-\mathbb{E}_{\nu_k}\left[M_k(T)\right])\label{eq:tmp1}.
\end{align}
The proof can now be concluded as follows. Since $J_k$ are disjoint, we have $\sum_{k\in[K]}N_k(T)\le T$ and thus there eixts $k^*$ such that $\mathbb{E}_{\nu_0}[N_{k^*}(T)]\le T/K$ and, since $M_k(T)\le N_k(T)$, also $\mathbb{E}_{\nu_0}[M_{k^*}(T)]\le T/K$ and by Markov $\mathbb{P}_{\nu_0}[M_{k^*}(T)\ge T/2]\le 2/K$. Now, by Pinsker's inequality and \Cref{claim:1}, we have
\[
|\mathbb{P}_{\nu_0}(A)-\mathbb{P}_{\nu_k}(A)|\le \sqrt{\frac{1}{2}\textup{KL}(\nu_0,\nu_k)}\le 10^2\epsilon\sqrt{\mathbb{E}_{\nu_0}[N_k(T)]},
\]
for any measurable event $A$ and instance $k\in[K]$. We can apply this to the event $\{M_{k^*}(T)\ge T/2\}$, which gives
\begin{align}\label{eq:tmp2}
\mathbb{P}_{\nu_{k^*}}[M_{k^*}\le T/2]\ge 1-\frac{2}{K}-\frac{1}{2}\epsilon\sqrt{\mathbb{E}_{\nu_0}[N_{k^*}(T)]}\ge 1-\frac2K-10^2\epsilon\sqrt{\frac TK},
\end{align}
where the last inequality comes from our choice of $k^*$. From, reverse Markov's inequality it is clear that $\mathbb{E}_{\nu_{k^*}}[M_{k^*}(T)]\le T-\frac{T}{2}\mathbb{P}_{\nu_{k^*}}[M_{k^*}(T)\le T/2]$, which plugged into \Cref{eq:tmp1}, gives
\begin{align*}
    \opt-\mathbb{E}_{\nu_{k^*}}\left[\sum_{t=1}^TU^P(p_t)\right]&\ge \frac{\epsilon T}{4}\mathbb{P}_{\nu_{k^*}}[M_{k^*}(T)\le T/2]\\
    &\ge \frac{\epsilon T}{4}\left(1-\frac2K-10^2\epsilon\sqrt{\frac TK}\right)\tag{\Cref{eq:tmp2}}.
\end{align*}
Now we can take $K=\lfloor\frac1{15\epsilon}\rfloor\le \frac{1}{15\epsilon}$, which is allowed by \Cref{lem:LBtmp1}. This gives:
\[
\opt-\mathbb{E}_{\nu_{k^*}}\left[\sum_{t=1}^TU^P(p_t)\right]\ge \frac{\epsilon T}{4}\left(1-32\epsilon-10^2\epsilon\sqrt{32T\epsilon}\right),
\]
choosing $\epsilon={T}^{-1/3}/2^7$ gives
\[
\opt-\mathbb{E}_{\nu_{k^*}}\left[\sum_{t=1}^TU^P(p_t)\right]\ge \frac{T^{2/3}}{2^{7}}-\frac{T^{1/3}}{2^{11}}-\frac{T^{2/3}}{2^{10}}\ge 
\frac{T^{2/3}}{200},
\]
which proves the desired lower bound. 
\end{proof}

\section{Proof of \Cref{thm:ub2}}\label{app:proof-ub-fixed}

Now, we focus on the setting where we want to minimize the regret against a single unknown type $\bar \theta \in [0,1]$. When $F_{a_2}^\top r \ge F_{a_1}^\top r$ there is a trivial solution that yields regret $0$ (\Cref{lemma:easy-case}). Therefore, in the following we focus only on the case in which $F_{a_1}^\top r > F_{a_2}^\top r$.

The section is structured as follows:
\begin{itemize}
    \item First, we discuss the properties of contracts played by \Cref{alg:e2c}. In particular, we first formally define $\lazy$, and then we show that contracts played by \Cref{alg:e2c} are optimal contracts that can be used to approximate $\bar \theta$. More precisely, we formally show that, for each $\theta \le \lazy$, there exists optimal contracts such that for all types smaller or equal than $\theta$, the best-response is $a_1$, while for types larger than $\theta$ the best response is $a_2$. Furthermore, we also show that, for types larger than $\lazy$, it is optimal to play the null contract. Finally, we conclude these sections by showing that these contracts have instantaneous which is bounded by $\mathcal{O}(\DeltaF)$.  
    \item Then, in \Cref{app:subsec-stat-test}, we provide proofs regarding the statistical tests that we use in \Cref{alg:e2c} to approximate correctly the hidden type $\bar \theta$.
    \item Finally, in \Cref{app:subsec-regret-anal-fixed-types}, we combine these results to analyze the regret of \Cref{alg:e2c} and prove \Cref{thm:ub2}.
\end{itemize}

\subsection{Properties of contracts played by \Cref{alg:e2c}}\label{app:subsec-lazy}
We start by presenting the formal definition of $\lazy$. Recall that we introduced $\lazy$ as the a type such that, for all $\theta > \lazy$, the optimal action to be induced is $a_2$, and hence, the optimal contract is the null contract $\bar p = (0, \dots, 0)$. The rationale behind $\lazy$ is that incentivizing $a_1$ for types larger than $\lazy$ is not convenient for the principal. Specifically, we define $\lazy$ as the value of the following linear program. 

\begin{equation}
\label{eq:theta-lazy}
\lazy \coloneqq
\left\{
\begin{array}{ll}
\displaystyle \max_{\theta \in [0,1], p \in [0,1]^m} & \theta \\[0.5em]
\text{s.t.} & (F_{a_1}-F_{a_2})^\top p \ge  \theta C_{a_1} \\[0.5em]
\phantom{\text{s.t.}} & F_{a_1}^\top (r-p) \ge F_{a_2}^\top r
\end{array}
\right.
\end{equation}

First of all, we observe that \Cref{eq:theta-lazy} is always well-defined since it has a non-empty and bounded feasible region. To see that this region is non-empty it is sufficient to consider $\theta = 0$ and $p=(0, \dots, 0)$. Then, all constraints are satisfied since $F_{a_1}^\top r > F_{a_2}^\top r$ holds by assumption. Furthermore, it is also easy to see that $\lazy > 0$, as $F_{a_1}^\top r$ is strictly greater than $F_{a_2}^\top r$ and $C_{a_1}$ is strictly greater than $0$. 

Next, we observe that for all $(\theta,p)$ in the feasible region, $(F_{a_1} - F_{a_2})^\top p \ge \theta C_{a_1}$, so that $a_1 \in \mathcal{B}(\theta,p)$, and, furthermore, $F_{a_1}^\top (r-p) \ge F_{a_2}^\top r$, so that the principal utility when playing $p$ is larger than the optimal utility that can be obtained when playing contracts that induce $a_2$, \ie $F_{a_2}^\top r$. Note, moreover, that $F_{a_1}^\top (r-p) \ge F_{a_2}^\top r \ge F_{a_2}^\top (r-p)$
This ensures that, when playing $p$ against $\theta$, $b(\theta,p)=a_1$. It follows directly from this definition that, for all $\theta \le \lazy$, it is optimal for the principal to induce $a_1$, while for all types greater than $\lazy$, it is optimal to induce $a_2$ using the null contract. This result is proved in the following lemma.

\begin{lemma}[Agent behaviors and optimal contracts]\label{lemma:agent-behav}
    For all $\theta > \lazy$, $p = (0, \dots, 0)$ is an optimal contract. Furthermore, for all $\theta \le \lazy$, there exists a contract $p^\star_{\theta}$ such that:
    \begin{align*}
    & \text{(i)~} \max_{p \in [0,1]^m} U^P(p, b(\theta, p)) =  U^P(p^\star_\theta, b(\theta, p^\star_\theta)) \\
    & \text{(ii)~} a_1 = b(\tilde \theta, p^\star_{\theta}) \quad \forall \tilde \theta \le \theta  \\
    & \text{(iii)~} a_2 = b(\tilde \theta, p^\star_\theta) \quad \forall \tilde \theta > \theta. 
    \end{align*}
    In particular, $p^\star_\theta$ is the solution of \eqref{eq:f-lambda} with $\lambda$ set to $\theta C_{a_1}$.
\end{lemma}
\begin{proof}
    First, fix any $\theta > \lazy$ and observe that the null contract yields utility $F_{a_2}^\top r$ to the principal. Furthermore, it is also easy to see that the null contract is optimal among all the contracts for which the agent plays $a_2$.  Next, focus on contracts for which the agent plays $a_1$. Specifically, suppose that there exists $p \in [0,1]^m$  such that $a_1 = b(\theta, p)$ and $F_{a_1}^\top (r-p) \ge F_{a_2}^\top r$, \ie agent plays $a_1$ when offered $p$, and $p$ yields the same (or more) utility to the principal than the null contract. This however is impossible, since $a_1 = b(\theta, p)$ implies that $(F_{a_1} - F_{a_2})^\top p\ge \theta C_{a_1}$, and, therefore, since, by assumption $F_{a_1}^\top (r-p) \ge F_{a_2}^\top r$ holds, $(\theta,p)$ would belong to the feasible region of \eqref{eq:theta-lazy}. Then, since $\theta > \lazy$, $\lazy$ cannot be an optimal solution of \eqref{eq:theta-lazy}. Thus, we proved that there is no contract which induces the agent to play $a_1$ while guaranteeing more utility for the principal. This concludes the first part of the proof.

    Second, consider $\theta \le \lazy$. We start by proving (i). Observe that, due to \Cref{lemma:opt-rewrite}, the maximum over contracts is well-defined. 
    \begin{align*}
        \max_{p \in [0,1]^m} U^P(p, b(\theta, p)) &  = \max \left\{ \max_{p: a_1 = b(\theta,p)} U^P(p, b(\theta, p)), \max_{p: a_2 = b(\theta,p)} U^P(p, b(\theta, p)) \right\} \\
        & = \max \left\{ \max_{p: a_1 = b(\theta,p)} U^P(p, b(\theta, p)), F_{a_2}^\top r\right\} \\
        & = \max \left\{ \max_{p: (F_{a_1}-F_{a_2})^\top p \ge \theta C_{a_1}} F_{a_1}^\top (r-p), F_{a_2}^\top r\right\} \\
        & = \max \left\{ F_{a_1}^\top r - \min_{p: (F_{a_1} - F_{a_2})^\top p = \theta C_{a_1}} F_{a_1}^\top p, F_{a_2}^\top r\right\} \\
        & = \max \left\{ F_{a_1}^\top (r-p^{\star}_\theta), F_{a_2}^\top r\right\}.
    \end{align*}
    where the last equality follows directly from the definition of $p^\star_\theta$. Now, observe that, since $p^\star_\theta$ is optimal for $\theta$ among all the contracts for which $a_1 \in \mathcal{B}(\theta,p)$, and since $\theta \le \lazy$, it must hold that the pair $(\theta, p^\star_\theta)$ is a feasible solution of \eqref{eq:theta-lazy}, and, therefore, $F_{a_1}^\top (r-p^{\star}_\theta) \ge F_{a_2}^\top r$. As a consequence, we have proved (i). 

    To prove (ii), it is sufficient to observe that $a_1 = b(\theta, p^\star_\theta)$. Hence, it must hold that $a_1 = b(\tilde \theta, p^\star_\theta)$ for all $\tilde \theta < \theta$ as well.

    Finally, to prove (iii), we observe that, by definition $p^\star_\theta$ satisfies $(F_{a_1} - F_{a_2})^\top p^\star_\theta = \theta C_{a_1}$. Therefore, for all $\tilde \theta > \theta$, action $a_1$ cannot belong to $\mathcal{B}(\tilde \theta, p^\star_\theta)$. 
    This concludes the proof.
\end{proof}

Next, we show that the instantaneous regret for the contracts that we designed above is upper bounded by $\DeltaF$, and, therefore, they are suitable for exploration.

\begin{lemma}[Sub-optimality gaps]\label{lemma:cost-of-exploration}
    Suppose $F_{a_1}^\top r > F_{a_2}^\top r$. 
    Consider $\theta_1, \theta_2 \in [0,1]$ and let $p_1^\star, p_2^\star$ be optimal contracts for $\theta_1$ and $\theta_2$ respectively. Then, it holds that:
    \begin{align}\label{eq:robust-eq-1}
    U^P(p^\star_1, b(\theta_1,p_1^\star)) - U^P(p_2^\star, b(\theta_1,p_2^\star)) \le 2 \DeltaF.
    \end{align}
    
    Furthermore, let $\alpha > 0$ and define $\tilde p_2 = p^\star_2 + \alpha( r- p^\star_2)$. Then, it holds that:
    \begin{align}\label{eq:robust-eq-2}
    U^P(p^\star_1, b(\theta_1,p_1^\star)) - U^P(p_2^\star, b(\theta_1,\tilde p_2^\star)) \le 2\DeltaF + \alpha.    
    \end{align}

    Finally, let $\bar p = (0, \dots, 0)$. It holds that:
    \begin{align}\label{eq:robust-eq-3}
        U^P(p^\star_1, b(\theta_1, p_1^\star)) - U^P(\bar p, b(\theta_1, \bar p)) \le \DeltaF.
    \end{align}
\end{lemma}
\begin{proof}
    We start by proving \Cref{{eq:robust-eq-1}}.
    First, observe that $p_1^\star$ and $p_2^\star$ are well-defined due to \Cref{lemma:opt-rewrite}.

    Second, fix any contract $p \in [0,1]$ and any $\theta_1, \theta_2$. Denote by $a = b(\theta_1, p)$ and $\bar a = b(\theta_2,p)$. 
    Then, it holds that:
    
    \[
        U^P(p, b(\theta_1, p)) - U^P(p, b(\theta_2, p)) =
        \begin{cases}
        0 & \text{if } a = a \\
        F_{a}^\top (r-p) - F_{\bar a}^\top (r-p) & \text{if } a \ne \bar a ,
        \end{cases}
    \]
    from which it directly follows that:
    \begin{align}\label{eq:delta-same-contract}
    U^P(p, b(\theta_1, p)) - U^P(p, b(\theta_2, p)) \le \DeltaF   . 
    \end{align}

    At this point, we are ready to prove our result.
    \begin{align*}
        U^P(p^\star_1, b(\theta_1,p_1^\star)) - U^P(p_2^\star, b(\theta_1,p_2^\star))  =  U^P(p^\star_1,& b(\theta_1,p_1^\star))  - U^P(p^\star_2, b(\theta_2, p^\star_2)) + \\& + U^P(p^\star_2, b(\theta_2, p^\star_2)) - U^P(p_2^\star, b(\theta_1,p_2^\star))
    \end{align*}
    Focus on the second term of the r.h.s. of this equation. Due to \Cref{eq:delta-same-contract}, it is smaller than $\DeltaF$. 
    Concerning the first term of the r.h.s., instead, we observe the following.
    \begin{align*}
    \max_{p \in [0,1]^m} U^P(p, b(\theta_1,p)) - \max_{p \in [0,1]^m} U^P(p, b(\theta_2,p)) & \le \max_{p \in [0,1]^m} ( U^P(p, b(\theta_1,p)) - U^P(p, b(\theta_2, p)) )  \\
    & \le \DeltaF. \tag{by \Cref{eq:delta-same-contract}}
    \end{align*}

    Combining these arguments, yields \Cref{{eq:robust-eq-1}}.
    
    We now continue by proving \Cref{{eq:robust-eq-2}}. To this end, we note that we can decompose $U^P(p^\star_1, b(\theta_1,p_1^\star)) - U^P(p_2^\star, b(\theta_1,\tilde p_2^\star))$ as the sum of $U^P(p^\star_1, b(\theta_1,p_1^\star)) - U^P(p_2^\star, b(\theta_1, p_2^\star))$ and $U^P(p_2^\star, b(\theta_1, p_2^\star)) - U^P(p_2^\star, b(\theta_1,\tilde p_2^\star))$. The first term can be bounded using \Cref{eq:robust-eq-1}, yielding the $\DeltaF$ term in \Cref{eq:robust-eq-2}. For the second term, instead, we apply \Cref{lemma:robust-contract} with $\epsilon = 0$ and $\rho(\theta) = b(\theta,p^*_2)$. This yields:
    \[
    U^P(p_2^\star, b(\theta_1, p_2^\star)) - U^P(p_2^\star, b(\theta_1,\tilde p_2^\star)) \le \alpha,
    \]
    which concludes the proof of \Cref{eq:robust-eq-2}.

    Finally, we prove \Cref{eq:robust-eq-3}. If $b(\theta_1, p^\star_1)=a_2$, then, the null contract must be an optimal contract for $\theta_1$ as all the other contracts that induce $a_2$ cannot yield larger utility to the principal. If $b(\theta_1, p^\star_1) = a_1$, instead, there are two cases. Either $b(\theta_1, \bar p)=a_1$, and again we have that the null contract must be an optimal contract for $\theta_1$, or $b(\theta_1, \bar p) = a_2$. In this case:
    \[
    F_{a_1}^\top (r-p^\star_1) - F_{a_2}^\top r \le (F_{a_1} - F_{a_2})^\top r \le \DeltaF,
    \]
    which concludes the proof.
\end{proof}

\subsection{Statistical Test}\label{app:subsec-stat-test}

In this section, we prove results regarding the statistical test that we use in \Cref{alg:e2c}. Specifically, \Cref{lemma:stat-test} shows how the $Z$ statistic works when we apply to a stream of $n$ i.i.d. outcomes realized from either $F_{a_1}$ or $F_{a_2}$. Then, in \Cref{lemma:correct-exploration}, we apply such tests within \Cref{alg:e2c}.

\begin{lemma}[Statistical Test]\label{lemma:stat-test}
    Consider $\{\omega_i\}_{i=1}^n$ be a stream of $n$ i.i.d. samples from $F_{a_2}$. Then, it holds that:
    \begin{align}\label{eq:test-1}
        \mathbb{P}\left( \frac{1}{n} \sum_{i=1}^n \mathds{1}\{ \omega_i \in \Omega_1 \} \ge \tau \right) \le \exp\left( -\frac{1}{8} n {\DeltaF^2} \right).
    \end{align}
    Similarly, if $\{\omega_i\}_{i=1}^n$ is a stream of $n$ i.i.d. samples from $F_{a_1}$, it holds that:
    \begin{align}\label{eq:test-2}
        \mathbb{P}\left( \frac{1}{n} \sum_{i=1}^n \mathds{1}\{ \omega_i \in \Omega_1 \} \le \tau \right) \le \exp\left( -\frac{1}{8} n {\DeltaF^2}\right).
    \end{align}
\end{lemma}
\begin{proof}
    We first prove \Cref{{eq:test-1}}. For brevity, we use  $Z$ to denote $\frac{1}{n} \sum_{i=1}^n \mathds{1}\{ \omega_i \in \Omega_1 \}$. Observe that 
    \[
    \mathbb{E}[Z] = \frac{1}{n} \sum_{i=1}^n \E\left[ \mathds{1}  \{ \omega \in \Omega_1 \}  \right] = \mathbb{P}(\omega \in \Omega_1) = \sum_{\omega \in \Omega_1} F_{a_2, \omega}. 
    \]
    Then, we have that:
    \begin{align*}
        \mathbb{P}\left( Z \ge \tau \right) & = \mathbb{P}\left( Z -\sum_{\omega \in \Omega_1} F_{a_2, \omega} \ge \tau-\sum_{\omega \in \Omega_1} F_{a_2, \omega} \right) \\
        & = \mathbb{P}\left( Z -\sum_{\omega \in \Omega_1} F_{a_2, \omega} \ge \frac{\sum_{\omega \in \Omega_1} (F_{a_1, \omega} - F_{a_2, \omega})}{2} \right) \tag{Def. of $\tau$}\\
        & = \mathbb{P}\left( Z -\sum_{\omega \in \Omega_1} F_{a_2, \omega} \ge \frac{\DeltaF}{4} \right)  \\
        & \le \exp \left( - \frac{1}{8} n {\DeltaF^2} \right) \tag{\Cref{lemma:hoeffding}}
    \end{align*}
    The proof of \Cref{eq:test-2} follows from similar arguments.
\end{proof}

\begin{lemma}[Correct Exploration]\label{lemma:correct-exploration}
    Consider \Cref{alg:e2c}. Suppose $\DeltaF > \beta$ and  $\bar \theta > \lazy$. It holds that:
    \begin{align}\label{eq:exp-algo-eq1}
        \Prob\left( Z_0 \ge \tau \right) \le \exp\left(-\frac{1}{8} n \DeltaF^2 \right).
    \end{align}
    Suppose that $\bar \theta \le \lazy$. It holds that:
    \begin{align}\label{eq:exp-algo-eq2}
        \Prob \left( \bar \theta \notin [s_{L}, e_L] \right) \le 2L \exp \left(-\frac{1}{8}n \DeltaF^2 \right).
    \end{align}
\end{lemma}
\begin{proof}
    The lemma is a combination of \Cref{lemma:agent-behav} and \Cref{lemma:stat-test}.

    We start by proving \Cref{eq:exp-algo-eq1}. If $\bar \theta > \lazy$, we know from \Cref{lemma:agent-behav}, that $\{ \omega^{(0)}_i \}_{i=1}^n$ is a stream of $n$ i.i.d. samples from $F_{a_2}$. Thus, applying \Cref{eq:test-1} in \Cref{lemma:stat-test} yields \Cref{eq:exp-algo-eq1}. 

    Second, we continue by proving \Cref{eq:exp-algo-eq2}. For each phase $l \in \{0, \dots, L-1 \}$, we denote by $\mathcal{E}_l$ the following event:
    \begin{align*}
    \mathcal{E}_l & = \left\{ \left( \{ \omega_i^{(l)}\}_{i=1}^n \sim F_{a_1} \land Z_l > \tau \}  \right) \lor \left( \{ \omega_i^{(l)}\}_{i=1}^n \sim F_{a_2} \land Z_l \le \tau \} \right)  \right\}.
    \end{align*}
    Now, observe that, under $\bigcap_{l=0}^{L-1} \mathcal{E}_l$, we are guaranteed that $\bar \theta \in [s_{L}, e_L]$. Indeed, if $\mathcal{E}_0$ holds, \Cref{alg:e2c} enters the exploration phase. Indeed, due to \Cref{lemma:agent-behav}, $\bar \theta \le \lazy$ responds to $p^\star_{\lazy}$ by playing $a_1$ and $Z_0 > \tau$ holds from $\mathcal{E}_0$. Then, the exploration phase runs a binary search sub-routine. Let us analyze the first phase $l=1$. Suppose that $\bar \theta \in \left(\tfrac{1}{2} \lazy , \lazy\right]$. Then, by \Cref{lemma:agent-behav}, $\bar \theta$ responds to $p^\star_{\tfrac{1}{2}\lazy}$ by playing $a_2$, and since $Z_1 \le \tau$ holds, we have that the search continues by inspecting the interval $[s_2, b_2] = \left[\tfrac{1}{2} \lazy , \lazy\right]$ which is the one that contains $\bar \theta$. Similarly, if $\bar \theta \in [0, \lazy]$, from \Cref{lemma:agent-behav}, we have that $\bar \theta$ plays $a_1$, and since $Z_1 > \tau$, we have that  the search continues by inspecting the interval $[s_2, b_2] = \left[0, \tfrac{1}{2} \lazy \right]$ which is the one that contains $\bar \theta$. Then, by an inductive argument, we obtain that $\bar \theta \in [s_L, e_L]$ under $\bigcap_{l=0}^{L-1} \mathcal{E}_l$. Then, once we reach the last phase we have
    \begin{align*}
        \Prob \left( \bar \theta \notin [s_{L}, e_L] \right) 
        & \le \sum_{i=0}^{L-1} \Prob \left( \mathcal{E}_l^\complement \right) \\
        & = \sum_{i=0}^{L-1} \Prob\left( \left\{ \{ \omega_i^{(l)} \}_{i=1}^n \sim F_{a_1} \land Z_l \le \tau \} \right\} \lor  \left\{\{ \omega_i^{(l)} \}_{i=1}^n \sim F_{a_2} \land Z_l > \tau \right\} \right)   \\
        & \le \sum_{i=0}^{L-1} \Prob\left( \left\{ \{ \omega_i^{(l)} \}_{i=1}^n \sim F_{a_1} \land Z_l \le \tau \} \right\} \right) + \Prob\left( \left\{ \{ \omega_i^{(l)} \}_{i=1}^n \sim F_{a_2} \land Z_l > \tau \} \right\} \right)  \\
        & \le \sum_{i=0}^{L-1} \Prob\left(  Z_l \le \tau  \Big| \{ \omega_i^{(l)} \}_{i=1}^n \sim F_{a_1} \right) +  \Prob\left(  Z_l \ge \tau  \Big| \{ \omega_i^{(l)} \}_{i=1}^n \sim F_{a_2} \right) \\
        & \le 2L \exp \left( -\frac{1}{8} n \DeltaF^2 \right),\tag{\Cref{lemma:stat-test}}
    \end{align*}
    which concludes the proof.
\end{proof}

\subsection{Regret Analysis}\label{app:subsec-regret-anal-fixed-types}

Finally, we analyze the regret of \Cref{alg:e2c}. We split the analysis in three cases. First, in \Cref{subsubsec:lazy} and \Cref{subsubsec:non-lazy} we analyze the regret for $\bar \theta > \lazy$ and $\bar \theta \le \lazy$ respectively. Then, in \Cref{subsubsec:sim-action} we analyze the regret for the case where $\DeltaF \le \beta$. Finally, in \Cref{subsubsec:combine} we combine these results and prove \Cref{thm:ub2} by properly choosing $n, L, \alpha$ and $\beta$. 

\subsubsection{Regret Analysis for Lazy Types}\label{subsubsec:lazy}

\begin{lemma}[Regret lazy types]\label{lemma:regret-lazy-types}
    Suppose that $\bar \theta > \lazy$ and $\DeltaF > \beta$. Then, \Cref{alg:e2c} guarantees that:
    \begin{align*}
        R_T \le 2 \DeltaF n + 2T (\DeltaF + \alpha) \exp\left( -\frac{1}{8}n\DeltaF^2\right).
    \end{align*}
\end{lemma}
\begin{proof}
    Denote by $p^\star$ the null contract. Recall that, from \Cref{lemma:agent-behav} the null contract is optimal for types greater than $\lazy$. Furthermore, since $\theta > \lazy$, it must be that $\lazy < 1$. Hence, it holds that \Cref{alg:e2c} tests whether $\bar \theta$ is lazy. It follows that:
    \begin{align*}
        R_T & = \sum_{t=1}^n \left( U^P(p^\star, \bar \theta) - U^P(p^\star_{\lazy}, \bar \theta) \right) + \sum_{t={n+1}}^T  \left( U^P(p^\star, \bar \theta) - \E [ U^P(p_t, \bar \theta)] \right) \\
        & \le 2 \DeltaF n + \sum_{t={n+1}}^T  \left( U^P(p^\star, \bar \theta) - \E [ U^P(p_t, \bar \theta)] \right) \tag{\Cref{lemma:cost-of-exploration}} \\
    \end{align*}
    Next, fix any $t \ge n+1$ and analyze $U^P(p^\star, \bar \theta) - \E [ U^P(p_t, \bar \theta)]$ using the law of total expectation. Denote by $\mathcal{E}_0 = \{ Z_0 \le \tau \}$. Observe that, under $\mathcal{E}_0$, it holds
    \begin{align*}
        U^P(p^\star, \bar \theta) - \E [ U^P(p_t, \bar \theta)\,\, | \,\,\mathcal{E}_0] = 0,
    \end{align*}
    as \Cref{alg:e2c} plays the optimal contract until the end of the horizon. Instead, when $\mathcal{E}_0^\complement$ holds, we split the analysis in two cases. First, consider the case when $t \le \{n+1, \dots, nL \}$. In this case, \Cref{alg:e2c} is playing $p^\star_{\theta}$ for some $\theta \in [0, \lazy]$ where $p^\star_{\theta}$ is as in \Cref{lemma:agent-behav}. Thus, it holds that:
    \begin{align*}
        U^P(p^\star, \bar \theta) - \E [ U^P(p_t, \bar \theta) | \mathcal{E}_0^\complement] & \le \sup_{\theta \in [0, \lazy]} U^P(p^\star, \bar \theta) - U^P(p^\star_\theta, \bar \theta)\\
        & \le 2 \DeltaF \tag{\Cref{lemma:cost-of-exploration}}
    \end{align*}
    If, instead, $t \ge nL + 1$, \Cref{alg:e2c} plays $\tilde p = p^\star_{\theta} + \alpha (r-p^\star_{\theta})$ for some $\theta \in [0, \lazy]$.
    It follows that:
    \begin{align*}
        U^P(p^\star, \bar \theta) - \E [ U^P(p_t, \bar \theta) | \mathcal{E}_0^\complement] & \le \sup_{\theta \in [0, \lazy]} U^P(p^\star, \bar \theta) - U^P(p^\star_\theta + \alpha(r-p^\star_\theta), \bar \theta) \\
        & \le 2 \DeltaF + \alpha. \tag{\Cref{lemma:cost-of-exploration}}
    \end{align*}

    Combining these results, we have that:
    \[
    \sum_{t={n+1}}^T  \left( U^P(p^\star, \bar \theta) - \E [ U^P(p_t, \bar \theta)] \right) \le 2T (\DeltaF + \alpha) \Prob\left(\mathcal{E}_0^\complement \right).
    \]
    Using \Cref{lemma:correct-exploration} we can upper bound $\Prob(\mathcal{E}_0^\complement)$, obtaining:
    \[
    R_T \le 2 \DeltaF n + 2T (\DeltaF + \alpha) \exp\left( -\frac{1}{8}n\DeltaF^2\right),
    \]
    which concludes the proof.
\end{proof}

\subsubsection{Regret Analysis for Non-Lazy Types}\label{subsubsec:non-lazy}

\begin{lemma}[Regret for non-lazy types]\label{lemma:regret-non-lazy}
Suppose that $\DeltaF > \beta$ and $\bar \theta \le \lazy$. It holds that:
    \[
    R_T \le 2nL \DeltaF + 2T \left(  \left( \frac{2^{-L}}{\alpha} + \alpha \right) + 2(\DeltaF + \alpha) L\exp\left( - \frac{n \DeltaF^2}{8} \right)   \right).
    \]
\end{lemma}
\begin{proof}
    Denote by $p^\star$ the optimal contract for the unknown true type $\bar \theta$. It holds that:
    \[
    R_T = \sum_{t=1}^{nL} \left( U^P(p^\star, \bar \theta) - \E [ U^P(p_t, \bar \theta)] \right) + \sum_{t={nL+1}}^T  \left( U^P(p^\star, \bar \theta) - \E [ U^P(p_t, \bar \theta)] \right).
    \]
    Focus on the first $nL$ steps and observes that, by definition, \Cref{alg:e2c} plays an optimal contract for some type $\theta \in [0,1]$. Thus, we can apply \Cref{lemma:cost-of-exploration} to obtain:
    \[
    \sum_{t=1}^{nL} \left( U^P(p^\star, \bar \theta) - \E [ U^P(p_t, \bar \theta)] \right) \le 2nL \DeltaF.
    \]
    Next, fix any $t \ge n+1$ and analyze $U^P(p^\star, \bar \theta) - \E [ U^P(p_t, \bar \theta)]$ using the law of total expectation. Precisely, we will use as event $\mathcal{E} = \{ \bar \theta \in [s_L, e_L] \}$. 
    
    We first analyze the case when $\mathcal{E}$ holds. In this case, we first observe that $|\bar \theta - \hat \theta_L| \le \frac{e_L - s_L}{2} \le 2^{-L}$. 
    
    Furthermore, \Cref{alg:e2c} plays $\tilde p = p^\star_{\hat \theta_L} + \alpha (r-p^\star_{\hat \theta_L})$ and we have that:
    \begin{align*}
    U^P(\tilde p, \bar \theta) &  \ge U^P(p^\star_{\hat \theta_L}, b(\hat \theta_L, p^\star_{\hat \theta_L}) - \left( \frac{2^{-L} C_{a_1}}{\alpha} + \alpha \right) \tag{\Cref{lemma:robust-contract}}
    \\ 
    & = U^P(p^\star_{\hat \theta_L}, \hat \theta_L) - \left( \frac{2^{-L} C_{a_1}}{\alpha} + \alpha \right) \\ & 
    \ge U^P(p^\star + \alpha( r-p^\star) , \hat \theta_L)- \left( \frac{2^{-L} C_{a_1}}{\alpha} + \alpha \right) \tag{$p^\star_{\hat \theta_L}$ is opt. for $\hat \theta_L$} \\
    & \ge U^P(p^\star, b(\bar \theta, p^\star)) - 2 \left( \frac{2^{-L} C_{a_1}}{\alpha} + \alpha \right) \tag{\Cref{lemma:robust-contract}}
    \end{align*}
    In particular, when we used \Cref{lemma:robust-contract}, we have used that types $\theta_1, \theta_2$ for which $|\theta_1 - \theta_2| \le \epsilon$ satisfy $b(\theta_1, p) \in \mathcal{B}_{C_{a_1} \cdot \epsilon}(\theta_2, p)$ for any contract $p$. Thus, under $\mathcal{E}$, we obtain that:
    \[
    U^P(p^\star, \bar \theta) - \E [ U^P(p_t, \bar \theta) | \mathcal{E}]  \le 2 \left( \frac{2^{-L} C_{a_1}}{\alpha} + \alpha \right) \le 2 \left(  \frac{2^{-L}}{\alpha} + \alpha\right).
    \]
    
    Next, suppose that $\mathcal{E}$ does not hold. Then, we split the analysis under two events $\mathcal{E}_1^\complement$ and $\mathcal{E}_2^\complement$ which are disjoint and that are such that $\mathcal{E}^\complement = \mathcal{E}_1^\complement \cup \mathcal{E}_2^\complement$. The first one, $\mathcal{E}_1^\complement$ is when \Cref{alg:e2c} plays the null contract as the first test on lazy types failed. In this case, \Cref{lemma:cost-of-exploration} yields:
    \[
    U^P(p^\star, \bar \theta) - \E [ U^P(p_t, \bar \theta) | \mathcal{E}_1^\complement]  \le 2 \DeltaF.
    \]
    In the second case, \ie on $\mathcal{E}_2^\complement$, instead, we only know that \Cref{alg:e2c} plays $\tilde p = p^\star_{\theta} + \alpha (r-p^\star_{\theta})$ for some $\theta \in [0, \lazy]$.
    Hence, we have that:
    \begin{align*}
        U^P(p^\star, \bar \theta) - \E [ U^P(p_t, \bar \theta) | \mathcal{E}_2^\complement] & \le \sup_{\theta \in [0, \lazy]} U^P(p^\star, \bar \theta) - U^P(p^\star_\theta + \alpha(r-p^\star_\theta), \bar \theta) \\
        & \le 2 \DeltaF + \alpha. \tag{\Cref{lemma:cost-of-exploration}}
    \end{align*}

    Combining these results, we have that, for any $t \ge nL+1$:
    \begin{align*}
    \left( U^P(p^\star, \bar \theta) - \E [ U^P(p_t, \bar \theta)] \right) & \le  2 \left( \frac{2^{-L}}{\alpha} + \alpha \right) + (2\DeltaF + \alpha) \left( \Prob(\mathcal{E}^\complement_1) + \Prob(\mathcal{E}^\complement_2) \right) \\
    & = 2 \left( \frac{2^{-L} }{\alpha} + \alpha \right) + (2\DeltaF + \alpha)  \Prob(\mathcal{E}^\complement)
    \end{align*}

    Combining these results with \Cref{lemma:correct-exploration}, we have that:
    \[
    R_T \le 2nL \DeltaF + 2T \left(  \left( \frac{2^{-L}}{\alpha} + \alpha \right) + 2(\DeltaF + \alpha) L\exp\left( - \frac{n \DeltaF^2}{8} \right)   \right),
    \]
    which concludes the proof.
\end{proof}

\subsubsection{Regret when Actions are Similar}\label{subsubsec:sim-action}

\begin{lemma}[Regret when $F$'s are close]\label{lemma:regret-close}
Suppose $\DeltaF \le \beta$. It holds that $R_T \le 2 T \beta$.
\end{lemma}
\begin{proof}
    When $\DeltaF \le \beta$, \Cref{alg:e2c} simply plays the null contract until the end of the horizon. The result is then a straightforward application of \Cref{lemma:cost-of-exploration}.
\end{proof}

\subsubsection{Combining the Results}\label{subsubsec:combine}

We are now ready to combine all the results to prove \Cref{thm:ub2}.

\begin{proof}[Proof of \Cref{thm:ub2}]
    We start by combining \Cref{lemma:regret-lazy-types}, \Cref{lemma:regret-non-lazy} and \Cref{lemma:regret-close}. Observe that the upper bound on the regret in \Cref{lemma:regret-lazy-types} is always larger than the one of \Cref{lemma:regret-non-lazy} for any choice of the parameters $\alpha, L$ and $n$. Hence, it is sufficient to optimize the parameters to minimize the upper bound in \Cref{lemma:regret-non-lazy} and \Cref{lemma:regret-close}.

    We choose $\beta$, $\alpha, L$ and $n$ as follows:
    \begin{align}
        & \beta = \frac{1}{\sqrt{T}} \\
        & L = \min\left\{ T , \lceil \log_2(4T) \rceil \right\} \label{eq:l-choiche} \\
        & n = \min \left\{ \Big\lfloor \frac{T}{L} \Big\rfloor, \max \left\{1,   \Big\lfloor \frac{8}{\DeltaF^2} \log\left( T \DeltaF^2 \right) \Big\rfloor  \right\} \right\} \label{eq:n-choiche} \\
        & \alpha = \min\left\{1,  \sqrt{2^{-L}} \right\} \label{eq:alpha-choiche} 
    \end{align}

    With this choice, we first observe that when $\DeltaF \le \frac{1}{\sqrt{T}}$, \Cref{lemma:regret-close} directly leads to $R_T \le \sqrt{T}$. Thus, we now focus on $\DeltaF > \frac{1}{\sqrt{T}}$ and we prove that \Cref{lemma:regret-non-lazy} yields a $\widetilde{\mathcal{O}}(\sqrt{T})$ bound for the regret, thus concluding the proof. We report here the result of \Cref{lemma:regret-non-lazy} for completeness. We have that $ R_T \le h_1 + h_2 + h_3 + h_4 $, where the different terms are as follows
    \begin{align*}
        & h_1 = 2nL \DeltaF, \quad  h_2 = 2T \left( \frac{2^{-L}}{\alpha} + \alpha \right) \\
        & h_3 = 4T \DeltaF L\exp\left( - \frac{n \DeltaF^2}{8} \right) \\
        & h_4 = 4T \alpha L\exp\left( - \frac{n \DeltaF^2}{8} \right)
    \end{align*}

    We proceed by showing that each $h_i$ belongs to $\widetilde{\mathcal{O}}(\sqrt{T})$.

    \paragraph{Upper bound on $h_1$}

    Suppose that $n = \lfloor T/L \rfloor$, which implies that 
    \[
    T \le L + L\max \left\{ {1, \Big\lfloor \frac{8}{\DeltaF^2} \log\left( T \DeltaF^2 \right) \Big\rfloor } \right\}.
    \]
    Furthermore, observe that this choice of $n$ ensures that $h_1 \le 2 T \DeltaF$.  Since $L \in \mathcal{O}(\log(T)$) and $\DeltaF \in [T^{-1/2}, 1]$, we have that $h_1 \in \widetilde{\mathcal{O}}(\sqrt{T})$.
    Next, suppose that $n = 1$. Here, from $L \in \mathcal{O}(\log(T))$ we obtain $h_1 \le L \in \widetilde{O}(\sqrt{T})$.
    Finally, suppose that $n = \Big\lfloor \frac{8}{\DeltaF^2} \log\left( T \DeltaF^2 \right) \Big\rfloor$. In this case,
    \[
    h_1  \in \mathcal{O}\left(L \frac{1}{\DeltaF} \log(T)\right) \le \mathcal{O}\left(L \sqrt{T} \log(T)\right) 
    \]
    Since $L \in \mathcal{O}(\log(T))$, $h_1 \in \widetilde{O}(\sqrt{T})$.

    \paragraph{Upper bound on $h_2$}
    We proceed by cases. Suppose that $\alpha = 1$ and $L = T$. In this case, $\alpha \le  \sqrt{2^{-L}}$, and we have that $h_2 \le 4 T 2^{-T} \in \widetilde{\mathcal{O}}(\sqrt{T})$.
    Next, suppose that $\alpha = 1$ and $L = \lceil \log_2(4T)\rceil$. Here, we have that $h_2 \le 4T  2^{-\lceil \log_2(4T)\rceil} \in \widetilde{\mathcal{O}}(\sqrt{T})$.
    Next, when $\alpha = \sqrt{2^{-L}}$ and $L = T$, $h_2 = 4T \sqrt{2^{-T}} \in \widetilde{\mathcal{O}}(\sqrt{T})$. 
    Finally, when $\alpha = \sqrt{2^{-L}}$ and $L = \lceil \log_2(4T)\rceil$, $h_2 = 4T \sqrt{2^{-\lceil \log_2(4T)\rceil} } \in \widetilde{\mathcal{O}}(\sqrt{T})$.

    \paragraph{Upper bound on $h_3$}
    First, suppose that  $n = \lfloor T/L \rfloor$, which implies that 
    \[
    T \le L + L\max \left\{ {1, \Big\lfloor \frac{8}{\DeltaF^2} \log\left( T \DeltaF^2 \right) \Big\rfloor } \right\}.
    \]
    We observe that $h_3$ reduces to:
    \begin{align*}
    h_3 & = 2  T \DeltaF L \exp \left( -\frac{T \DeltaF^2}{8L} \right) & \\
    & \le 2  T \DeltaF L \exp \left( -\frac{1}{8L} \right)      \tag{$\DeltaF \ge 1/\sqrt{T}$} \\
    & \le 2T \DeltaF L \\
    & \le 2 \DeltaF L \left( L + L\max \left\{ {1, \Big\lfloor \frac{8}{\DeltaF^2} \log\left( T \DeltaF^2 \right) \Big\rfloor } \right\} \right) \\
    & \le 2 L^2 + 2L^2\max \left\{ {1, \Big\lfloor \frac{8}{\DeltaF} \log\left( T \DeltaF^2 \right) \Big\rfloor } \right\} \\
    & \le 2 L^2 + 2L^2\max \left\{ {1, \Big\lfloor 8 \sqrt{T} \log\left( T  \right) \Big\rfloor } \right\} \tag{$\DeltaF \ge 1 / \sqrt{T}$}
    \end{align*}
    Since $L \in \mathcal{O}(\log(T))$, we obtained that $h_3 \in \widetilde{\mathcal{O}}(\sqrt{T})$.

    Next, suppose that $n=1$. In this case, $1 \ge \frac{8}{\DeltaF^2} \log\left( T\DeltaF^2 \right)$, which implies that 
    \[
    T \DeltaF \le \frac{\exp(\DeltaF^2 / 8)}{\DeltaF} \le \frac{1}{\DeltaF}.
    \]
     Thus, we have that:
    \begin{align*}
        h_3 & = 2T \DeltaF L \exp \left( -\frac{n \DeltaF^2}{8} \right) \\
        & \le  2 T \DeltaF L  \\
        & \le \frac{2L}{\DeltaF} \\
        & \le 2L \sqrt{T} \tag{$\DeltaF \ge  1 / \sqrt{T}$}
    \end{align*}
    Since $L \in \mathcal{O}(\log(T))$, we have $h_3 \in \widetilde{\mathcal{O}}(\sqrt{T})$.

    Finally, consider the last case where $n$  is attained at $\Big\lfloor \frac{8}{\DeltaF^2} \log\left( T \DeltaF^2 \right) \Big\rfloor$. In this case:
    \begin{align*}
        h_3 & = 2 T \DeltaF L \exp \left( - \log(T\DeltaF^2) \right) \\
        & = 2 T \DeltaF L \frac{1}{T \DeltaF^2} \\
        & = \frac{2L}{\DeltaF} \\
        & \le 2L \sqrt{T}.\tag{$\DeltaF \ge  1 / \sqrt{T}$}
    \end{align*}
    Since $L \in \mathcal{O}(\log(T))$, we have $h_3 \in \widetilde{\mathcal{O}}(\sqrt{T})$.

    \paragraph{Upper bound on $h_4$}
    Here, we aim to upper bound $\alpha$ with $\DeltaF$. This allows to directly exploit the result for $h_3$ to show that this term belongs to $\widetilde{\mathcal{O}}(\sqrt{T})$.
    To this end, observe that $\alpha^2 \le 2^{-L} \le \DeltaF$ holds if and only if $L \ge \log_2(1 / \DeltaF^2)$. Since $\DeltaF \ge 1/\sqrt{T}$, however, we have that this is equivalent to $L \ge \log_2(T)$, which is always true for our choice of $L$.

     This concludes the proof.
\end{proof}

\section{Proof of \Cref{th:LB_single}}\label{app:lbsingle}

\begin{proof}[Proof of \Cref{th:LB_single}]
    We build two instances $\I_0$ and $\I_1$ as follows. Both instances have two actions with $c_{a_1}=1$ and $c_{a_2}=0$. Moreover, in both instances, we have two outcomes with rewards $r_{\omega_1}=1$ and $r_{\omega_2}=0$. The outcome distribution for $a_1$ is $F_{a_1,\omega_1}=\frac12+4\epsilon$ and $F_{a_1,\omega_2}=\frac12-4\epsilon$, while for $a_2$ are $F_{a_2,\omega_1}=F_{a_2,\omega_2}=\frac12$. The two instances only differ by the type $\theta$ and in particular, in the first instance $\theta_0=\frac{40\epsilon}{\frac1\epsilon+8}$ and $\theta_1=\frac{24\epsilon}{\frac1\epsilon+8}$. Note that in both instances it is optimal to put a zero payment on the second outcome $\omega_2$. Let $x=p_{\omega_1}$ be the payment on the first outcome.

    Note that an agent of type $\theta$ plays action $a_1$ if and only if $(\frac12+4\epsilon)x-\theta\ge \frac x2$, and thus if and only if $x\ge \frac{\theta}{4\epsilon}$. Thus, in instance $\I_0$, the principal utility and mean of the Bernoulli describing the feedback distribution $\tau(x)$ (of observing outcome $\omega_1$) are:
    
    \begin{minipage}{.47\textwidth}
    \begin{align*}
        U_0^P(x)=\begin{cases}
            \frac12(1-x)&\text{if}\quad x< \frac{10}{\frac1\epsilon+8}\\
            (\frac12+4\epsilon)(1-x)&\text{otherwise}
        \end{cases}
    \end{align*}
    \end{minipage}\hfill\text{and}\hfill
    \begin{minipage}{.47\textwidth}
    \begin{align*}\tau_0(x)=
        \begin{cases}
            \frac12&\text{if}\quad x< \frac{10}{\frac1\epsilon+8}\\
            \frac12+4\epsilon&\text{otherwise}
        \end{cases},
    \end{align*}
    \end{minipage}

    while for instance $\I_1$ are:

    \begin{minipage}{.47\textwidth}
    \begin{align*}
        U_1^P(x)=\begin{cases}
            \frac12(1-x)&\text{if}\quad x<\frac{6}{\frac1\epsilon+8}\\
            (\frac12+4\epsilon)(1-x)&\text{otherwise}
        \end{cases}
    \end{align*}
    \end{minipage}\hfill\text{and}\hfill
    \begin{minipage}{.47\textwidth}
    \begin{align*}\tau_1(x)=
        \begin{cases}
            \frac12&\text{if}\quad x< \frac{6}{\frac1\epsilon+8}\\
            \frac12+4\epsilon&\text{otherwise}
        \end{cases}.
    \end{align*}
    \end{minipage}

    It is clear that the $x\in\{0,\frac{6}{\frac1\epsilon+8},\frac{10}{\frac1\epsilon+8}\}$ are dominating actions in all instances in terms of both feedback and rewards. Then, a standard change-of-measure argument allows us to conclude the lower bound. Define $N_0(T)$ the number of times $x_0=0$ was played up to time $T$, and $N_1(T)$ and $N_2(T)$ the number of times $x_1=\frac{6}{\frac1\epsilon+8}$ and $x_2=\frac{10}{\frac1\epsilon+8}$ was played. Then, there exists an absolute constant $C>0$ such that:
    \[
    \textnormal{KL}(\nu_0,\nu_1)
    \le C\mathbb{E}_{\nu_0}[N_1(T)]\epsilon^2,
    \]
    moreover, direct computations show that for some constant $c>0$
    \[
    U_0^P(x_0)-U_0^P(x_1)\ge c\epsilon, \quad U_1^P(x_1)-U_1^P(x_0)\ge c\epsilon, \quad\text{and}\quad U_1^P(x_1)-U_1^P(x_2)\ge c\epsilon.
    \]
    Hence
    \[
    \E_{\nu_0}[R_T]\ge c\epsilon\E_{\nu_0}[N_1(T)],
    \qquad
    \E_{\nu_1}[R_T]\ge c\epsilon\E_{\nu_1}[T-N_1(T)].
    \]
    Now, if $\mathbb{E}_{\nu_0}[N_1(T)]\ge \tfrac T4$ then $\E_{\nu_0}[R_T]\ge \Omega(\epsilon T)$. On the other hand, if $\mathbb{E}_{\nu_0}[N_1(T)]< \tfrac T4$ then Markov's inequality gives:
    \[
    \mathbb{P}_{\nu_0}\left(N_1(T)>\frac T2\right)\le \frac12,
    \]
    and (with $\epsilon=\Theta(1/\sqrt{T})$), Pinsker's inequality gives \begin{align*}
        \mathbb{P}_{\nu_1}\left(N_1(T)>\frac T2\right)&\le \mathbb{P}_{\nu_0}\left(N_1(T)>\frac T2\right)+TV(\nu_0,\nu_1)\\
        &\le \frac34
    \end{align*} 
    which can be rewritten as $\mathbb{P}_{\nu_1}(N_1(T)\le T/2) \ge 1/4$. Finally, again by Markov's inequality:
    \begin{align*}
    \mathbb{E}_{\nu_1}[R_T]&\ge c\epsilon\mathbb{E}_{\nu_1}[T-N_1(T)]\\
    &\ge c\epsilon\frac{T}{2} \mathbb{P}_{\nu_1}\Big[T-N_1(T)\ge \frac{T}2\Big] \\
    &=\Omega\Big(\epsilon T \mathbb{P}_{\nu_1}\Big[N_1(T)\le \frac{T}2\Big]\Big)\\
    &\ge \Omega(\epsilon T),
    \end{align*}
    concluding the proof.
%
\end{proof}

\section{Useful Results}\label{app:others}

\subsection{The $F_{a_2}^\top r \ge F_{a_1}^\top r$ case}

In this section, we prove that when $F_{a_2}^\top r \ge F_{a_1}^\top r$ holds, there exists a trivial algorithm that yields $R_T = 0$ for any sequence $\{ \theta_t \}_{t=1}^T$.

\begin{lemma}[The $F_{a_2}^\top r \ge F_{a_1}^\top r$ case]\label{lemma:easy-case}
    Suppose that $F_{a_2}^\top r \ge F_{a_1}^\top r$. Then, there exists an algorithm for which $R_T = 0$. 
\end{lemma}
\begin{proof}
    Let $\bar p = (0, \dots, 0)$ and observe that, for any $\theta \in [0,1]$, we have that 
    \[
    \argmax_{a \in \{ 1,2\}} F_{a}^\top \bar p - \theta \cdot C_a = -\argmin_{a \in \{ 1,2\}} \theta \cdot C_a.
    \]
    Hence, since $C_{a_2} \le C_{a_1}$, $a_2 \in \mathcal{B}(\theta, \bar p)$ for all $\theta \in [0,1]$. Furthermore, since $F_{a_2}^\top r \ge F_{a_1}^\top r$, it also holds that $a_2 \in b(\theta, \bar p)$ for all types $\theta$. As a consequence, we have that $U^P(\bar p, b(\theta,\bar p )) = F_{a_2}^\top r$ for all $\theta \in [0,1]$. Next, we observe that:
    \begin{align*}
        \opt & = \sup_{p \in [0,1]^m} \left( \sum_{t=1}^T \mathds{1} \{ a_1 \in b(\theta_t, p)   \} F_{a_1}^\top (r-p) + \mathds{1} \{ b(\theta_t, p) = a_2  \} F_{a_2}^\top (r-p) \right) \\
        & \le \sup_{p \in [0,1]^m} \left( \sum_{t=1}^T \mathds{1} \{ a_1 \in b(\theta_t, p)   \} F_{a_1}^\top r + \mathds{1} \{ b(\theta_t, p) = a_2  \} F_{a_2}^\top r  \right) \tag{$p \ge 0$}\\
        &  \le  \sum_{t=1}^T F_{a_2}^\top r  \tag{$F_{a_2}^\top r \ge F_{a_1}^\top r$} \\
        & = \sum_{t=1}^T U^P(\theta_t, b(\theta_t, \bar p)) \tag{$U^P(\bar p, b(\theta,\bar p )) = F_{a_2}^\top r$ for all $\theta \in [0,1]$}.
    \end{align*}
    Hence, the constant strategy $\bar p$ is optimal for any sequence $\{ \theta_t \}_{t=1}^T$ thus concluding the proof.  
\end{proof}

\subsection{Concentration Inequality}

We report for completeness a statement of the Hoeffding's inequality for the concentration of i.i.d. random variables with values in $[0,1]$. 

\begin{lemma}[Hoeffding's inequality \citep{boucheron2003concentration}]\label{lemma:hoeffding}
    Consider $X_1, \dots, X_n$ be $n$ i.i.d. random variables in $[0,1]$ whose expected value is $\mu$. Then, for $\delta > 0$, it holds that:
    \begin{align*}
        \mathbb{P}\left( \frac{1}{n} \sum_{i=1}^n X_i - \mu \ge \delta \right) \le \exp\left( - {2n \delta^2} \right).
    \end{align*}
\end{lemma}

\subsection{Contract Linearization}

\begin{lemma}[Proposition 2.4. in \cite{dutting2021complexity}]\label{lemma:robust-contract}
    Let $\epsilon > 0$ and  $\rho: [0,1] \to \{a_1, a_2 \}$ be a generic mappings from types to actions, and let $p \in [0,1]^m$. Suppose that $\rho(\theta) \in \mathcal{B}_\epsilon(\theta,p)$ holds for all $\theta$. Then, let $\tilde p = p + \alpha(r-p)$ for some $\alpha > 0$. It holds that:
    \[
    U^P(\tilde p, b(\theta, \tilde p)) \ge U^P(p, \rho(\theta)) - \left( \frac{\epsilon}{\alpha} + \alpha \right).
    \]
\end{lemma}

\subsection{Optimization of Upper Semicontinuous Functions}

We recall the definition of upper semicontinuity. Assume that $(X, d)$ is a metric space and consider $f: X \to \Reals$. Then, $f$ is upper semicontinuous at $x_0 \in X$ if for every $\epsilon > 0$ there exists $\delta > 0$ such that $f(x) < f(x_0) + \epsilon$ for all $x$ such that $d(x, x_0) < \delta$. We now state a well-known result for optimization of upper semicontinuous functions.

\begin{theorem}\label{thm:opt-usc}
    If $C$ is a compact set, and $f: C \to \Reals$ is upper semicontinuous, then $f$ attains a maximum on $C$.
\end{theorem}



\end{document}